\documentclass[a4paper,11pt,fleqn]{article}
\usepackage[T1]{fontenc}
\usepackage[margin=25mm]{geometry}
\usepackage[svgnames,dvipsnames]{xcolor}
\usepackage[british]{babel}
\usepackage{natbib}
\usepackage{microtype}
\usepackage{amsthm,mathtools,amssymb}
\usepackage{enumitem}
\usepackage{algorithm,algorithmic}
\usepackage{graphicx}
\usepackage{subcaption}
\usepackage{tikz}
\usetikzlibrary{automata,positioning,arrows.meta}

\newtheorem{theorem}{Theorem}[section]
\newtheorem{lemma}{Lemma}[section]
\newtheorem{proposition}{Proposition}[section]
\newtheorem{corollary}{Corollary}[section]
\newtheorem{definition}{Definition}[section]

\newtheorem{remark}{Remark}
\newtheorem{example}{Example}

\newcommand{\DI}{D_I}
\newcommand{\DO}{D_O}

\newcommand{\out}{\mathrm{out}}
\newcommand{\Lang}{\mathcal{L}}

\newcommand{\Nat}{\mathbb{N}}
\newcommand{\fun}{\longrightarrow}

\newcommand{\smta}{\textsc{Solve}}
\newcommand{\smsmt}{\textsc{SolveMaxSMT}}
\newcommand{\expr}{\textit{expr}}
\newcommand{\pwr}{\mathbb{P}}

\newcommand{\ttt}{\mathtt{true}}
\newcommand{\fff}{\mathtt{false}}
\newcommand{\ol}{\overline}

\newcommand{\TS}{\text{TS}}
\newcommand{\olg}{{\ol G}}

\usepackage[hidelinks]{hyperref}
\allowdisplaybreaks
\title{Testing and Learning Symbolic Finite State Machines}
\author{%
  Wen-ling Huang\thanks{Corresponding author. Email: \texttt{huang@uni-bremen.de}. ORCID: 0000-0002-9915-5357.}
  \and
  Jan Peleska\thanks{Email: \texttt{peleska@uni-bremen.de}. ORCID: 0000-0003-3667-9775.}}
\date{%
  \small University of Bremen, Department of Mathematics and Computer Science\\
  \small Bibliothekstrasse 1, 28359 Bremen, Germany}

\begin{document}
\maketitle

\begin{abstract}
Symbolic finite state machines (SFSMs) describe input/output behaviour using guards and output assignments with possibly infinite data domains. We study deterministic and completely specified SFSMs whose guards and output assignments depend only on the current input. We define finite representative input sets that contain witnesses for relevant guard overlaps and separating witnesses for output assignments that differ on those overlaps. Our main theorem shows that language equivalence of the finite instantiations implies language equivalence over the full input domain. This result transfers complete testing methods for deterministic finite state machines (DFSMs) to SFSMs, provided finite sets of admissible guards and output assignments and an upper bound on the number of distinguishable reachable states are known. 
Under these assumptions, a DFSM learner with complete testing can learn a finite instantiation, which is then lifted to an equivalent SFSM. We establish a bound on the size of representative input sets and give an SMT construction whose correctness and termination hold under stated solver assumptions.
\end{abstract}

\par

\section{Introduction}\label{sec:intro}
This article develops a unified framework for \emph{model-based testing} (MBT) and \emph{model learning} (ML) of   \emph{symbolic finite state machines} (SFSMs), based on finite representative input sets.

\paragraph{Background.}
 MBT checks whether a system under test (SUT) conforms to a given reference model. Relative to a specified fault domain and testing assumptions, a test suite is \emph{complete} if passing it is equivalent to conformance. Many complete testing methods have been developed for finite state machines (FSMs) with finite input and output alphabets, using language equivalence or language inclusion as conformance  relations~\cite{vasilevskii1973,chow:wmethod,luo_test_1994,hierons_testing_2004,petrenko_testing_2011,DBLP:conf/forte/DorofeevaEY05,SimaoPY09}. 
%DBLP:conf/pts/SimaoPY09
Model learning constructs a model of the observable SUT behaviour by applying input sequences and observing the outputs. Angluin's $L^*$-algorithm~\cite{Angluin1987} learns finite automata through membership and equivalence queries. Several algorithms adapt this approach to FSMs~\cite{Niese03, ShahbazG09,IsbernerHS14,DBLP:conf/tacas/VaandragerGRW22}. 

In our setting, complete FSM testing answers the equivalence queries. Each hypothesis serves as the reference model for test generation. A failed test provides a counterexample for the learner. If all tests pass, the hypothesis is accepted and learning terminates, provided that the upper state bound and the assumptions of the testing method hold. 

A learned model can be used to check requirements expressed, for example, in Linear Temporal Logic (LTL). If two models have the same language, then they satisfy the same LTL formulae over input/output observations. Once the model is known to be language equivalent to the SUT, model checking can establish whether the SUT satisfies such requirements. 

% Learning the model of the observable SUT behaviour is particularly attractive in the context of \emph{property checking}, where it has to be decided whether an SUT satisfies a property (i.e.~a requirement) specified in some logic such as Linear Temporal Logic (LTL). With a model reflecting the true SUT behaviour at hand, classical model checking algorithms can be used to verify whether the specified properties are satisfied by the SUT. For SUTs whose behaviour can be represented by FSMs, this so-called \emph{black-box checking} approach was first introduced by Peled et al.~\cite{JALC-2002-225}.
% Note that in this article, we do not develop SFSM property checking itself; rather, we focus on the MBT and learning foundations required for such applications.

% ...............................................................................
\paragraph{Motivation.}
 
SFSMs extend FSMs by replacing input symbols with guards over input variables and output symbols with assignments $y = e(x)$. Here $x$ and $y$ are tuples of input and output variables, and $e$ is a function. The data domains may be infinite, whereas classical FSM testing and learning require finite alphabets. In addition, the observed output $e(a)$ need not identify the assignment used by the SUT: different output assignments may have the same value at the selected input $a$. 

 Petrenko~\cite{DBLP:journals/sosym/Petrenko19} studied complete testing theories for SFSMs with symbolic inputs and outputs.\footnote{In his article, he used the term SIOFSM to distinguish symbolic finite state machines from less expressive machines introduced in his previous work. Complete test suites are called checking experiments.} He distinguished
an assignment/output fault domain from a more general transition fault domain. The latter permits arbitrary implementation guards but restricts output assignments to a given finite set. Implementations have
at most as many states as the reference model. For infinite input domains, completeness of an arbitrary concrete instantiation of his symbolic checking experiment also requires the $\Omega^d$-converter condition. It requires each reference-defined symbolic input sequence of length $d$ to produce a single symbolic output sequence in the implementation. We consider a different fault domain, defined by finite sets of admissible guards and output assignments. Our aim is to obtain complete concrete test suites without this condition and with an upper state bound that may exceed the number of reference states. We also use the finite input construction to support model learning.  

\paragraph{Main contributions.}

We define finite representative sets $\Sigma_I$	
that contain witnesses for the relevant guard overlaps and separating inputs for output assignments that differ on those overlaps.   Restricting an SFSM to $\Sigma_I$ gives a DFSM with finite input and output alphabets. Our main theorem shows that,  if $\Sigma_I$ is representative for two SFSMs, equivalence  on all input sequences from $\Sigma_I^*$ implies their language equivalence over the full concrete input domain. This result has two applications:
\begin{itemize}
    \item  For model-based testing, a test suite that is complete for the DFSM fault domain is also complete for the corresponding SFSM fault domain. Thus, any complete DFSM testing method whose assumptions hold can be used to generate executable SFSM tests.
    
    \item For model learning,  we construct a representative input set from a known input partition and a finite set of output assignments.
    An $L^*$-style DFSM learning algorithm learns the finite instantiation and complete DFSM testing answers its equivalence queries. We then lift the minimal DFSM hypothesis to an SFSM. Under the assumptions stated below, the lifted hypothesis is language equivalent to the SUT over the full input domain.
\end{itemize}

Our earlier work~\cite{DBLP:conf/fsen/HuangKP23,DBLP:conf/pts/BruningGHKPS23,DBLP:journals/scp/HuangKP24,zenodo-fsen-techreport-2022} investigated complete SFSM testing and property-oriented testing. 
The general construction selects representative inputs from I/O equivalence classes. The specialised approach uses an input partition and a separability condition on output expressions. The present construction can yield smaller representative input alphabets and may therefore reduce the resulting test suites.

% Compared with our earlier constructions~\cite{DBLP:conf/fsen/HuangKP23,DBLP:conf/pts/BruningGHKPS23,DBLP:journals/scp/HuangKP24,zenodo-fsen-techreport-2022}, the present construction can yield smaller representative input alphabets and may therefore reduce the resulting test suites.
We give an SMT-based construction of representative input sets, prove its correctness and termination, and bound the number of solver calls. Appendix~\ref{app:algorithms} gives specialised constructions, MaxSMT variants, and a greedy set-cover procedure for reducing the number of representatives.

\paragraph{Assumptions.}
 Throughout the paper, we consider deterministic and completely specified SFSMs whose guards and output assignments depend only on the current input. We assume that the SUT can be reset reliably to its initial state before each test case and membership query.

For model-based testing, finite sets $G$ and $E$ are given such that the SUT has an SFSM representation using only guards from $G$ and output assignments from $E$.  We assume
an upper bound $m$ on the number of distinguishable reachable states of this representation. The representation need not be unique, and the SUT need not use these guards and output assignments in its implementation. 

For model learning, static analysis is assumed to provide a finite set of branching conditions and a finite set containing all output assignments of an SFSM representation of the SUT. The input partition induced by the branching conditions must refine every guard occurring in that representation. We also require an upper bound on its number of distinguishable reachable states. This bound  may be obtained by static analysis or abstract interpretation. The solver assumptions used to construct representative inputs are stated in Section~\ref{ssec:construction-assumptions} and Appendix~\ref{ssec:smt}.
 
\paragraph{Overview.}
Section~\ref{sec:prelim} introduces the SFSM model, its concrete semantics, finite instantiations, and output distinguishability. Section~\ref{sec:ris} defines general, model-specific, and pair-specific representative input sets. Section~\ref{sec:mb} proves the main equivalence theorem and its application to complete testing. Section~\ref{sec:learning} presents the learning method and the lifting of DFSM hypotheses to SFSMs. Section~\ref{sec:algorithms} gives an SMT-based algorithm for computing representative input sets and proves its correctness and termination. Detailed variants appear in Appendix~\ref{app:algorithms}. Section~\ref{sec:related} discusses related work, and Section~\ref{sec:conc} concludes the paper.

% =================================================================================
\section{Preliminaries}\label{sec:prelim}
This section introduces the SFSM model and notation used in the paper.

\paragraph{Symbolic finite state machines.}
A \emph{symbolic finite state machine} is a tuple
\[
M=(Q,q_0,V,D,G,E,h),
\]
where $Q$ is a finite, non-empty set of states and $q_0\in Q$ is the initial state. The finite variable set $V=V_I\cup V_O$ consists of input variables $V_I$ and output variables $V_O$, with $V_I\cap V_O=\emptyset$. Each variable $v\in V$ has a domain $D_v$, and $D$ is the collection of these domains. 

Fix orderings of input and output variables and write 
%$V_I=\{\chi_1,\dots,\chi_k\}$ and $V_O=\{o_1,\dots,o_\ell\}$ for input and output variables $\chi_i$ and $o_j$, respectively. We write   
\[x=(\chi_1,\dots,\chi_k),\qquad y=(o_1,\dots,o_\ell).\] 
The concrete input and output domains are
\[
\DI := D_{\chi_1}\times\cdots\times D_{\chi_k},
\qquad
\DO := D_{o_1}\times\cdots\times D_{o_\ell}.
\]
These domains may be infinite. 

Let $G$ be a finite set of guards over the input variables. The set $E$ is a finite set of output assignments of the form $y=e(x)$, where $e$ is a function $e : D_I\fun D_O$. Since the syntactic form of output assignments is always $y = e(x)$, we write $e$ for the assignment $y=e(x)$, so $E = \{ e_1,\dots, e_p \}$ denotes the corresponding set of output functions. 
The transition relation is
\[
h\subseteq Q\times G\times E\times Q .
\]
We write $q\xrightarrow{g/e}q'$ for a transition $(q,g,e,q')\in h$.

A concrete input $a=(a_1,\dots,a_k)\in\DI$ is a \emph{witness} (also called a \emph{model}) for a guard $g\in G$ 
if $g$ evaluates to true when each $\chi_i$ is assigned the value $a_i$. We use the same symbol $g$ for the formula and its set of satisfying inputs, and write $a \in g$ when $a$ satisfies $g$. Under this convention, conjunction $g_1\wedge g_2$ corresponds to intersection  $g_1\cap g_2$ and disjunction $g_1\vee g_2$ corresponds to union $g_1\cup g_2$.  The inclusion $g_1\subseteq g_2$ means that $g_1\implies g_2$ holds for every concrete input.
We omit transitions with unsatisfiable guards, since no concrete input can trigger them.

An SFSM $M$ is \emph{deterministic} if no concrete input enables two distinct transitions from the same state. Formally, for any two transitions $(q,g_i,e_i,q_i')\in h$, $i=1,2$, and any $a\in\DI$, the condition $a\in g_1\cap g_2$ implies $g_1=g_2$, $e_1=e_2$, and $q_1'=q_2'$. It is \emph{completely specified} if, for every $q\in Q$ and every $a\in\DI$, there is a transition $(q,g,e,q')\in h$ with $a\in g$. Throughout the paper, we consider deterministic and completely specified SFSMs. 

For $q\in Q$, let
\[
G(q)=\{g\in G\mid \exists (e,q')\in E\times Q:(q,g,e,q')\in h\}
\]
be the set of guards on transitions leaving $q$. The set of guards occurring in the transition relation is
\[
G_h=\bigcup_{q\in Q}G(q).
\]
The \emph{symbolic input/output alphabet} of $M$ is its set of transition labels,
\[
A_M=\{(g,e)\in G\times E\mid \exists q,q'\in Q:(q,g,e,q')\in h\}.
\]

For $(q,g,e,q')\in h$, define the symbolic output by $\out(q,g)=e$. For a concrete input $a\in g$, define $\out(q,a)=e(a)$. Thus, $\out(q,g)$ is an output assignment and $\out(q,a)$ is a concrete output value. The symbolic output is  well defined for every $g\in G(q)$, 
% \[
% \out : \{(q,g)\mid q\in Q,\ g\in G(q)\}\to E,
% \]
and the concrete output function is 
\[
\out : Q\times\DI\to\DO; (q,a)\mapsto \out(q,g)(a), \ \text{where $g$ is the unique guard in $G(q)$ containing $a$.} 
\]

\paragraph{Concrete semantics.}
For a set $A$, let $A^*$ denote the set of finite sequences over $A$, including the empty sequence $\varepsilon$. 
A symbolic input/output sequence
\[
g_1/e_1.\ g_2/e_2.\ \dots.\ g_k/e_k
\]
is executable from state $q\in Q$ if there exist transitions
\[
(q,g_1,e_1,q_1),\ (q_1,g_2,e_2,q_2), \dots,\ (q_{k-1},g_k,e_k,q_k)\in h .
\]
For each concrete input sequence $\alpha=a_1\dots a_k\in \DI^*$ with $a_i\in g_i$ for all $1\le i\le k$, the corresponding  concrete trace is
\[
a_1/e_1(a_1).\ a_2/e_2(a_2).\ \dots.\ a_k/e_k(a_k).
\]
We regard $a/b$ as the input/output pair $(a,b)$, so concrete traces belong to $(\DI\times D_O)^*$. The extension of $\out$ to finite sequences is 
\[\out(q,\alpha)=e_1(a_1).e_2(a_2).\ \dots.\ e_k(a_k),\qquad \out(q,\varepsilon)=\varepsilon.\]
The \emph{language of a state} is the set of all concrete traces executable from that state, including the empty trace.  The \emph{language of $M$}, denoted $\Lang(M)$, is the language of its initial state. Two states are \emph{equivalent} if their languages coincide, and \emph{distinguishable} otherwise. A state is \emph{reachable} if some concrete input sequence takes the machine from $q_0$ to that state.  The number of \emph{distinguishable reachable states} is the number of language-equivalence classes among the reachable states. An SFSM is \emph{reduced} if every state is reachable and no two distinct states are equivalent.

\paragraph{Output distinguishability.}
Let $X\subseteq \DI$. Two output assignments $e$ and $e'$ are \emph{distinguishable on $X$} if there exists a concrete input $a\in X$ such that $e(a)\neq e'(a)$. Such an input is  a \emph{separating witness} for $e$ and $e'$ on $X$. Otherwise, $e$ and $e'$ are \emph{equivalent on $X$}, denoted by $e\equiv_X e'$. Thus, 
\[e\equiv_X e' \iff \forall a\in X: e(a)=e'(a). \]

\paragraph{Finite instantiations.}
Let $M=(Q,q_0,V,D,G,E,h)$ be an SFSM, and let $\Sigma_I\subseteq\DI$ be a non-empty finite set. The \emph{finite instantiation of $M$ over $\Sigma_I$} is the DFSM
\[
M|_{\Sigma_I}=(Q,q_0,\Sigma_I,\Sigma_O,h|_{\Sigma_I}),
\]
with output alphabet
\[
\Sigma_O=\{e(a)\mid a\in\Sigma_I,\ (g,e)\in A_M,\ a\in g\}.
\] 
For $q,q'\in Q$, $a\in\Sigma_I$, and $b\in \Sigma_O$, its transitions are defined by 
\[
(q,a,b,q')\in h|_{\Sigma_I}
\quad\Longleftrightarrow\quad
\exists (g,e)\in A_M: (q,g,e,q')\in h\wedge a\in g\wedge b=e(a).
\]
The output alphabet $\Sigma_O$ is finite because $\Sigma_I$ and $A_M$ are finite. Since each state and concrete input enable exactly one symbolic transition, $M|_{\Sigma_I}$ is deterministic and completely specified. The construction retains the state set $Q$; it does not remove unreachable states or merge equivalent states. 

The finite instantiation preserves exactly the behaviour of $M$ on input sequences from $\Sigma_I^*$:
\[\Lang(M|_{\Sigma_I})=\Lang(M)|_{\Sigma_I}=\Lang(M)\cap (\Sigma_I\times D_O)^*.\]
Here, $\Lang(M)|_{\Sigma_I}$ denotes the traces of $M$ whose input components belong to $\Sigma_I$. Thus, the finite instantiation is an \emph{under-approximation} of the full concrete language: $\Lang(M|_{\Sigma_I}) \subseteq \Lang(M)$. It contains no trace with inputs outside $\Sigma_I$. 
%-------------------------------------------------------
\begin{example}[A deterministic and completely specified SFSM]\label{ex:sfsm}
Consider the SFSM $M=(Q,q_0,V,D,G,E,h)$ with $Q=\{q_0,q_1,q_2,q_3\}$, initial state $q_0$, input variable $x$, output variable $y$, and domain $\DI = \DO=\mathbb{Z}$. The guard set is 
$G=\{g_1,\ldots,g_6\}$, where
\[
\begin{aligned}
&g_1\equiv(x\le0), &&g_2\equiv(x>0),\\
&g_3\equiv(x\equiv0\pmod 2), &&g_4\equiv(x\equiv1\pmod 2),\\
&g_5\equiv(x\in(-10,10)), &&g_6\equiv(x\notin(-10,10)).
\end{aligned}
\]
The output assignment set is $E = \{ e_1,\dots,e_5 \}$ with
\[
e_1\equiv (y=0),\quad e_2\equiv (y=x),\quad e_3\equiv (y=x^2),\quad e_4\equiv (y=2x),\quad e_5\equiv (y=x+1).
\]
The transition relation $h$ consists of
\[
\begin{array}{rclcrcl}
q_0 &\xrightarrow{g_1\,/\,e_1}& q_0, &\qquad& q_0 &\xrightarrow{g_2\,/\,e_4}& q_1,\\[2pt]
q_1 &\xrightarrow{g_3\,/\,e_2}& q_2, &&
q_1 &\xrightarrow{g_4\,/\,e_5}& q_3,\\[2pt]
q_2 &\xrightarrow{g_3\,/\,e_3}& q_2, &&
q_2 &\xrightarrow{g_4\,/\,e_4}& q_3,\\[2pt]
q_3 &\xrightarrow{g_5\,/\,e_1}& q_2, &&
q_3 &\xrightarrow{g_6\,/\,e_1}& q_3.
\end{array}
\]
At each state, the two outgoing guards are disjoint and cover $\DI$. Hence $M$ is deterministic and completely specified. Figure~\ref{fig:sfsm} shows the machine. 
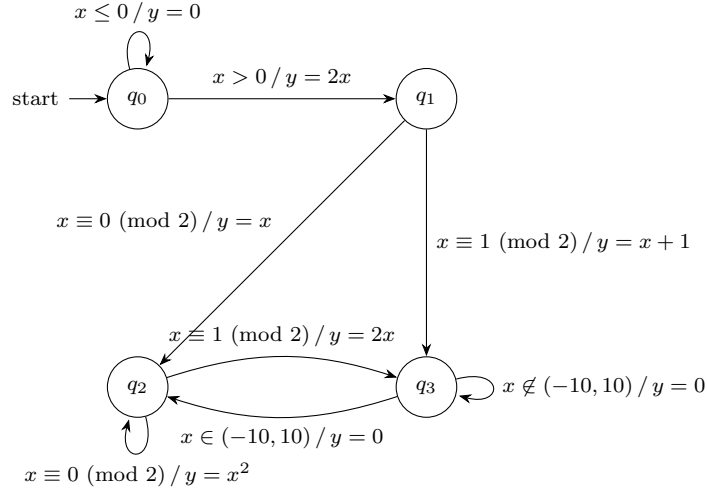
\begin{figure}
    \centering
    
\begin{center}
\begin{tikzpicture}[>=Stealth, node distance=30mm, font=\scriptsize, every state/.style={draw,minimum size=8mm}]
\node[state, initial] (q0) {$q_0$};
\node[state, right=of q0] (q1) {$q_1$};
\node[state, below=of q1] (q3) {$q_3$};
\node[state, left=of q3] (q2) {$q_2$};
\path[->]
(q0) edge[loop above] node{$x\le 0\,/\,y=0$} (q0)
(q0) edge node[above]{$x>0\,/\,y=2x$} (q1)
(q1) edge node[above left]{$x\equiv 0\ (\mathrm{mod}\ 2)\,/\,y=x$} (q2)
(q1) edge node[right]{$x\equiv 1\ (\mathrm{mod}\ 2)\,/\,y=x+1$} (q3)
(q2) edge[loop below] node{$x\equiv 0\ (\mathrm{mod}\ 2)\,/\,y=x^2$} (q2)
(q2) edge[bend left=18] node[midway, sloped, above, align=center]{$x\equiv 1\ (\mathrm{mod}\ 2)\,/\,y=2x$} (q3)
(q3) edge[bend left=18] node[midway, sloped, below, align=center]{$x\in (-10,10)\,/\,y=0$} (q2)
(q3) edge[loop right] node{$x\not\in (-10,10)\,/\,y=0$} (q3);
\end{tikzpicture}
\end{center}
    \caption{A deterministic and completely specified SFSM $M$.}
    \label{fig:sfsm}
\end{figure}
\end{example}

\begin{example}[An SFSM with two input and two output variables]\label{ex:sfsm-2in-2out}
Consider the SFSM $M=(Q,q_0,V,D,G,E,h)$ with $Q=\{q_0,q_1,q_2\}$, initial state $q_0$, input variables $u,v$, and output variables $y_1,y_2$. All variables have integer domains, so $\DI=\DO=\mathbb{Z}^2$. The guard set is $G=\{g_1,g_2,g_3,g_4\}$, where
\[
 g_1\equiv (u\le 0),\quad g_2\equiv (u>0),\quad
 g_3\equiv (v\le 0),\quad g_4\equiv (v>0).
\]
The output assignment set is $E=\{e_1,e_2,e_3,e_4\}$, where 
\[
e_1(u,v)=(0,0),\quad
e_2(u,v)=(u,v),\quad
e_3(u,v)=(u+v,u-v),\quad
e_4(u,v)=(2u,v+1).
\]
Each function $e_i$ defines the assignment $(y_1,y_2)=e_i(u,v)$.
% \[
% \begin{aligned}
% &e_1\equiv (y_1=0\wedge y_2=0),
% &&e_2\equiv (y_1=u\wedge y_2=v),\\
% &e_3\equiv (y_1=u+v\wedge y_2=u-v),
% &&e_4\equiv (y_1=2u\wedge y_2=v+1).
% \end{aligned}
% \]
The transition relation consists of 
\[
\begin{array}{rclcrcl}
q_0 &\xrightarrow{g_1/e_1}& q_0, &\qquad& q_0 &\xrightarrow{g_2/e_2}& q_1,\\[2pt]
q_1 &\xrightarrow{g_3/e_3}& q_2, && q_1 &\xrightarrow{g_4/e_4}& q_0,\\[2pt]
q_2 &\xrightarrow{g_1/e_1}& q_0, && q_2 &\xrightarrow{g_2/e_2}& q_2.
\end{array}
\]
Figure~\ref{fig:sfsm2} shows the machine.

The set $G$ does not partition $\DI=\mathbb{Z}^2\colon$ for example, $g_1$ overlaps with both $g_3$ and $g_4$. However, the guards of the outgoing transitions at each state are pairwise disjoint and cover $\DI\colon$ states $q_0$ and $q_2$ use guards $g_1$ and $g_2$, while $q_1$ uses $g_3$ and $g_4$. Hence, the machine is deterministic and completely specified.
\begin{figure}
    \centering    
%\begin{center}
\begin{tikzpicture}[>=Stealth, node distance=30mm, font=\scriptsize, every state/.style={draw,minimum size=8mm}]
\node[state, initial] (q0) {$q_0$};
\node[state, right=of q0] (q1) {$q_1$};
\node[state, below=of q1] (q2) {$q_2$};
\path[->]
(q0) edge[loop above] node{$u\le 0/y_1=0\wedge y_2=0$} (q0)
(q0) edge node[above]{$u>0/y_1=u\wedge y_2=v$} (q1)
(q1) edge[bend left=18] node[below]{$v\le 0/y_1=u+v\wedge y_2=u-v$} (q2)
(q1) edge[bend left=18] node[above]{$v>0/y_1=2u\wedge y_2=v+1$} (q0)
(q2) edge[bend left=18] node[left]{$u\le 0/y_1=0\wedge y_2=0$} (q0)
(q2) edge[loop right] node{$u>0/y_1=u\wedge y_2=v$} (q2);
\end{tikzpicture}
%\end{center}
%\includegraphics[width=0.5\linewidth]{}
    \caption{An SFSM with two input and two output variables.}
    \label{fig:sfsm2}
\end{figure}
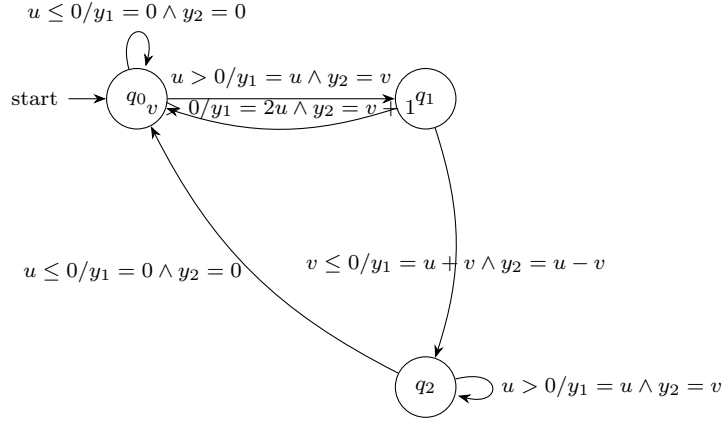
\end{example}

%---------------------------------------------------
\begin{table}[htbp]
    \centering
    \begin{tabular*}{\linewidth}
        {@{\extracolsep{\fill}}ll}
    \hline
    Symbol    & Meaning\\ \hline
   $M,M'$& Reference and implementation SFSMs in testing; $M$ is unknown in learning.\\ \hline
  $G, E$ & Finite sets of admissible guards and output assignments.\\ \hline
  $G_h, E_h, A_M$ & Guards, assignments and labels occurring in $M$.\\ \hline
  $G',\olg$ & Branching conditions and their non-empty input classes in learning.\\ \hline
  $\ol M$ & Refinement of the unknown $M$ to guards from $\olg$.\\ \hline
  $\Sigma_I,M|{\Sigma_I}$ &Finite concrete input set and the corresponding DFSM instantiation.\\ \hline
  $H_{\Sigma_I},H, \hat H$ &Learned minimal DFSM, lifted SFSM and coarsened SFSM.\\ \hline
  $N, E_{\varphi}$ &Minimal implementation quotient; assignment representatives on class $\varphi$ .\\ \hline
  $m,n$ &Upper bound on implementation states; reference state count as specified locally.\\ \hline
    \end{tabular*}
    \caption{Main notation.}
    \label{tab:notation}
\end{table}

% ====================================================================
\section{Representative input sets}\label{sec:ris}

Two SFSMs may have language-equivalent finite instantiations over a finite input set $\Sigma_I$ without being language-equivalent over the full input domain $\DI$. We define \emph{representative input sets $\Sigma_I$} by two requirements: witnesses for guard overlaps and separating witnesses for output assignments. Under these requirements, equivalence of the finite instantiations implies equivalence over $\DI$, as shown in Section~\ref{sec:mb}. Section~\ref{sec:learning} uses this result for model learning.  

\begin{definition}\label{def:repGE}
Let $G$ be a finite set of guards and let $E$ be a finite set of
output assignments. A finite input set $\Sigma_I \subseteq \DI$ is \emph{representative for $(G,E)$} if the following requirements hold for all $g,g'\in G$ and all $e,e'\in E$.
\begin{enumerate}
    \item \textit{Guard-overlap requirement.} If $g\cap g'\neq \emptyset$ then $g\cap g'\cap\Sigma_I\neq\emptyset$.
    \item \textit{Output-separation requirement.} If $e$ and $e'$ are distinguishable on $g\cap g'$ then there exists a separating input $a\in g\cap g'\cap \Sigma_I$ such that $e(a) \neq e'(a)$.
\end{enumerate}
\end{definition}
The case $g=g'$ is included, so every non-empty guard contains an input in $\Sigma_I$. This also witnesses satisfiability of each guard, consistently with omitting unsatisfiable transitions. Different pairs of output assignments may require different separating inputs.
Since $G$ and $E$ are finite, there are finitely many requirements. Choosing one witness for each applicable requirement gives a finite representative set. Section~\ref{sec:algorithms} presents an algorithm for constructing such sets. Further variants are in Appendix~\ref{app:algorithms}. 

%==========================================
\begin{proposition}
   Let $r=|G|, t=|E|$. A representative input set $\Sigma_I$ for $(G,E)$ exists with cardinality at most $\max\{\frac{r(r+1)}{2}, \frac{r(r+1)t(t-1)}{4}\} $. Moreover, consider the partition of $\DI\times \DO$ induced by $G$ and $E$.  Choose one concrete input-output pair from each class and let $R$ be the set of their input components. Then $R$ is representative for $(G,E)$. In particular, the minimum cardinality of a representative set for $(G,E)$ is at most $|R|$. 
\end{proposition}
\begin{proof}
    Since $|G|=r$ and $|E|=t$, there are $r(r+1)/2$ guard overlaps and $t(t-1)/2$ pairs of distinct output assignments. For each non-empty guard overlap, choose one separating witness for each distinguishable output-assignment pair. At most $t(t-1)/2$ witnesses are needed. If no pair is distinguishable on the current overlap, choose one witness from that overlap. Hence, there exists a representative set with at most $\max\{\frac{r(r+1)}{2}, \frac{r(r+1)t(t-1)}{4}\} $ inputs. 

    For the second claim, a witness $a$ of a guard overlap together with any concrete output  belongs to an I/O class. The selected I/O valuation from this class has an input component in the same guard overlap. If $e_p(a)\neq e_q(a)$, $(a,e_p(a))$ satisfies the output assignment $y=e_p(x)$ and violates $y=e_q(x)$. Every concrete I/O valuation in this I/O class has this property. Its selected input component therefore separates $e_p$ and $e_q$ on the overlap. Hence, $R$ satisfies both requirements of Definition~\ref{def:repGE}. 
\end{proof}
The I/O partition is defined in our earlier work~\cite{zenodo-fsen-techreport-2022}. Two valuations $(a,b),(a',b')\in \DI\times \DO$ are equivalent if they satisfy the same guards and output assignments. 
The following example illustrates how representative input sets contribute to uncovering errors in an SUT.  
\begin{example}\label{ex:rep}
Consider the SFSM $M$ given in Example~\ref{ex:sfsm} over $\DI=D_O=\mathbb{Z}$ with guard set $G=\{g_1,\dots, g_6\}$ and set of output assignments $E=\{e_1,\dots, e_5\}$.

Then $\Sigma_I=\{-11,-2,7,14\}$ is representative for $(G,E)$. Table~\ref{tab:set cover} lists the selected inputs in each guard overlap. Every non-empty guard overlap contains a witness from $\Sigma_I$, so the guard-overlap requirement is satisfied. Table~\ref{tab:out} shows that the five output assignments yield pairwise distinct values at each selected input. Thus every selected input in a guard overlap separates all distinct output assignments in $E$, and the output-separation requirement also holds.   

\begin{table}[htbp]
    \centering
    \begin{tabular*}{.65\linewidth}
        {@{\extracolsep{\fill}}|c|c|c|c|c|c|c|}
    \hline
        & $g_1$&$g_2$& $g_3$&$g_4$& $g_5$& $g_6$\\ \hline
   $g_1$& $-11$, $-2$&$\emptyset$& $-2$&$-11$& $-2$& $-11$\\ \hline
  $g_2$ & $\emptyset$&$7$,$14$& $14$&$7$& $7$ & $14$\\ \hline
  $g_3$ & $-2$&$14$& $-2$,$14$&$\emptyset$& $-2$& $14$\\ \hline
  $g_4$ & $-11$&$7$& $\emptyset$&$-11$, $7$& $7$& $-11$\\ \hline
  $g_5$ & $-2$&$7$& $-2$&$7$& $-2$,$7$& $\emptyset$\\ \hline
  $g_6$ & $-11$&$14$& $14$&$-11$& $\emptyset$& $-11$,$14$\\ \hline
    \end{tabular*}
    \caption{Selected inputs in each guard overlap. The entry in row $i$ and column $j$ is $g_i\cap g_j\cap \Sigma_I$.}
    \label{tab:set cover}
\end{table}

\begin{table}[htbp]
    \centering
    \begin{tabular*}{.65\linewidth}
        {@{\extracolsep{\fill}}|c|c|c|c|c|c|}
        \hline
        $a$ & $e_1(a)$ & $e_2(a)$ & $e_3(a)$ & $e_4(a)$ & $e_5(a)$ \\ \hline
        $-11$ & $0$ & $-11$ & $121$ & $-22$ & $-10$ \\ \hline
        $-2$  & $0$ & $-2$  & $4$   & $-4$  & $-1$  \\ \hline
        $7$   & $0$ & $7$   & $49$  & $14$  & $8$   \\ \hline
        $14$  & $0$ & $14$  & $196$ & $28$  & $15$  \\ \hline
    \end{tabular*}
    \caption{$\Sigma_I$ fulfils the output-separation requirement, since every $a\in \Sigma_I$ satisfies $e_p(a)\neq e_q(a)$, for $1\le p<q\le 5 $.}
    \label{tab:out}
\end{table}

The guards in $G$ define eight input equivalence classes, where two inputs are equivalent if they are in the same guards. Definition~\ref{def:repGE} requires witnesses for pairwise guard overlaps, not for each input equivalence class. The four concrete inputs of $\Sigma_I$ suffice.

Now suppose that the SUT is represented by an SFSM $M'=(Q',q_0',V,D,G,E,h')$ shown in Figure~\ref{fig:sut}, where $Q'=\{q_0', q_1', q_2', q_3'\}$. It differs from $M$ only in the guards of transitions leaving $q_1'$. The transitions $(q_1,g_3,e_2,q_2)$ and $(q_1,g_4,e_5,q_3)$ of $M$ are replaced by $(q_1',g_1,e_2,q_2')$ and $(q_1',g_2,e_5,q_3')$. 
All other transitions are the same as in $M$, with primed state names. Since $g_1$ and $g_2$ are disjoint and cover $\DI$,
$M'$ is deterministic and completely specified.

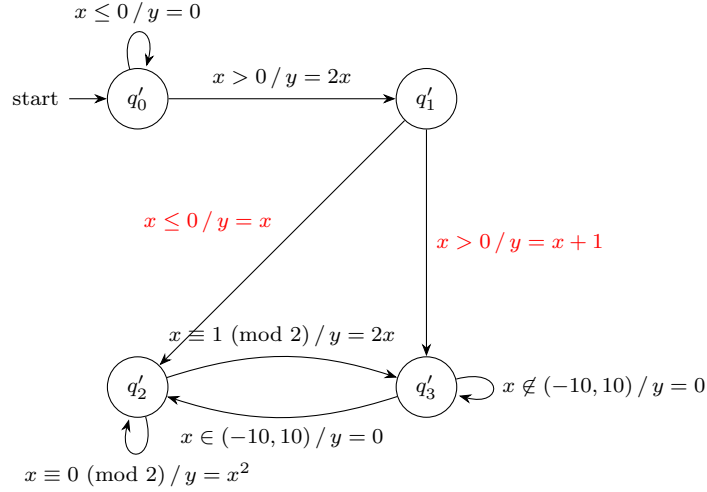
\begin{figure}[htbp]
    \centering    
%\begin{center}
\begin{tikzpicture}[>=Stealth, node distance=30mm, font=\scriptsize, every state/.style={draw,minimum size=8mm}]
\node[state, initial] (q0) {$q_0'$};
\node[state, right=of q0] (q1) {$q_1'$};
\node[state, below=of q1] (q3) {$q_3'$};
\node[state, left=of q3] (q2) {$q_2'$};
\path[->]
(q0) edge[loop above] node{$x\le 0\,/\,y=0$} (q0)
(q0) edge node[above]{$x>0\,/\,y=2x$} (q1)
(q1) edge node[above left]{\color{red}$x\le 0 \,/\,y=x$} (q2)
(q1) edge node[right]{\color{red}$x>0\,/\,y=x+1$} (q3)
(q2) edge[loop below] node{$x\equiv 0\ (\mathrm{mod}\ 2)\,/\,y=x^2$} (q2)
(q2) edge[bend left=18] node[midway, sloped, above, align=center]{$x\equiv 1\ (\mathrm{mod}\ 2)\,/\,y=2x$} (q3)
(q3) edge[bend left=18] node[midway, sloped, below, align=center]{$x\in (-10,10)\,/\,y=0$} (q2)
(q3) edge[loop right] node{$x\not\in (-10,10)\,/\,y=0$} (q3);
\end{tikzpicture}
%\end{center}
%\includegraphics[width=0.5\linewidth]{}
    \caption{An SFSM representation $M'$ of the SUT with two incorrect guards.}
    \label{fig:sut}
\end{figure}

To detect faults in $M'$, first apply the input $7\in g_2\cap \Sigma_I$ to reach state $q_1$ in $M$ and state $q_1'$ in $M'$. Both $M$ and $M'$ produce output $2\cdot 7=14$. Next, apply input $14\in g_3\cap g_2\cap \Sigma_I$.
In $M$, input $14$ triggers the transition labelled by $g3/e2$, producing output $e_2(14)=14$. In $M'$ it triggers the erroneous transition
labelled by $g2/e5$, producing output $e_5(14)=15$. Hence, 
the input sequence $7.14\in \Sigma_I^*$
detects a fault in $M'$.
\end{example}

%------------------------------------------------------
% representative for (M,G,E)
For MBT, the reference model $M$ is known. We assume that the SUT has an SFSM representation $M'$ whose
transitions use only guards from a given set $G$ and output assignments from a given set $E$. The labels of $M$ need not belong to $G\times E$.
 
\begin{definition}\label{def:repMGE}
 
Let $M$ be an SFSM, $G_h$ the set of guards occurring in $M$ and $A_M$ its set of transition labels. Let $G$ and $E$ be finite sets of guards and output assignments, respectively, over the same variables and domains as $M$. 
A finite input set $\Sigma_I \subseteq \DI$ is \emph{representative for $(M,G,E)$} if 
\begin{enumerate}
\item The guard-overlap requirement holds for every $g\in G_h$ and every $g'\in G$.
\item The output-separation requirement holds for every $(g,e)\in A_M$ and every $(g',e')\in G\times E$.     
\end{enumerate}
\end{definition}
The \emph{model-specific} definition compares labels occurring in $M$ with possible SUT labels from $G\times E$.
Let
\[
E_h=\{e\mid \exists g\colon (g,e)\in A_M\}
\]
be the set of output assignments occurring in $M$. Every finite input
set that is representative for
$(G_h\cup G,E_h\cup E)$ is also representative for
$(M,G,E)$, but the converse need not hold. Thus, the minimum
cardinality of a representative set for $(M,G,E)$ is no greater than that 
for $(G_h\cup G,E_h\cup E)$, and it may be strictly smaller.

%----------------------------------------------------
% representative for (M,M')
Suppose that the transition labels of an SFSM $M'$ representing the SUT are also known. The \emph{pair-specific} definition then requires witnesses only between labels occurring in $M$ and  labels occurring in $M'$.  
\begin{definition}\label{def:repMM}
Let $M=(Q,q_0,V,D,G,E,h)$ and $M'=(Q',q_0',V,D,G',E',h')$ be two deterministic and completely specified SFSMs.
A finite input set $\Sigma_I \subseteq \DI$ is \emph{representative for $(M,M')$} if  
\begin{enumerate}
\item The guard-overlap requirement holds for every $(g, g')\in G_h\times G_{h'}'$.
\item The output-separation requirement holds for all  $(g,e)\in A_M$ and $(g',e')\in A_{M'}$.
\end{enumerate}
\end{definition}

% \begin{lemma}\label{rm:rep}
% Let $M=(Q,q_0,V,D,G,E,h)$ and $M'=(Q',q_0',V,D,G',E',h')$ be two deterministic and completely specified SFSMs. 
% Let $E_{h'}'$ be the set of output assignments occurring in $M'$.
% Let $G''$ and $E''$ be finite sets  of guards and output assignments, respectively, such that
% \[
% G_h\cup G_{h'}'\subseteq G'',\qquad E_h\cup E_{h'}'\subseteq E''.
% \]
% Every finite input set $\Sigma_I$ that is representative for $(G'',E'')$ or for $(M,G',E')$ is also representative for $(M,M')$. Thus,
% the minimum cardinality of a representative set for $(M,M')$ is no greater than that for $(M,G',E')$, and it may be strictly smaller.  
% \end{lemma}

\begin{lemma}\label{lemma:rep}
 Let $M$ be a deterministic and completely specified SFSM, and let $G$, $E$ be finite sets
of guards and output assignments over the same variables and domains as $M$. If
$\Sigma_I \subseteq D_I$ is representative for $(M,G,E)$ in the sense of Definition~\ref{def:repMGE}, then $\Sigma_I$ is
representative for $(M,M')$ in the sense of Definition~\ref{def:repMM}, for every deterministic and completely
specified SFSM $M'$ whose transitions use only guards from $G$ and output assignments from $E$.
Thus, the minimum cardinality of a representative set for $(M,M')$ is no greater than that for $(M,G,E)$, and it may be strictly smaller.  
\end{lemma}
\begin{proof}
Let $M'$ satisfy the stated conditions. Then $G'_{h'} \subseteq G$, and $A_{M'} \subseteq G \times E$. Each guard-overlap and output-separation requirement of Definition~\ref{def:repMM} is therefore a requirement of Definition~\ref{def:repMGE}. Hence, $\Sigma_I$ is a representative set for $(M,M')$.

% Let $(g,g') \in G_h \times G'_{h'}$ with $g \cap g' \neq \emptyset$. Since
% $g \in G_h$ and $g' \in G$, $(g,g') \in G_h \times G$, by the guard-overlap requirement of Definition~\ref{def:repMGE}, 
% $g \cap g' \cap\Sigma_I \neq \emptyset$.

% Let $(g,e) \in A_M$, $(g',e') \in A_{M'}$ and assume $e,e'$
% distinguishable on $g \cap g'$. As $(g',e') \in G \times E$, by the output-separation requirement of condition~2 of Definition~\ref{def:repMGE}, there exists some 
% $a \in g \cap g' \cap \Sigma_I$ with $e(a) \neq e'(a)$.
% Both requirements of Definition~\ref{def:repMM} hold.   
\end{proof}
 
\begin{remark}\label{rm:rep}
Let $G$, $E$ be finite sets
of guards and output assignments. 
Let $M$, $M'$ be two deterministic and completely specified SFSMs whose transitions use only guards from $G$ and output assignments from $E$. If
$\Sigma_I \subseteq D_I$ is representative for $(G,E)$ in the sense of Definition~\ref{def:repGE}, then $\Sigma_I$ is
representative for $(M,G,E)$ and $(M,M')$. 

The first implication follows because $G_h\subseteq G$ and $A_M\subseteq G\times E$; the second follows from Lemma~\ref{lemma:rep}.
\end{remark}

% ============================================================================ 
\section{Model-based testing}\label{sec:mb}
In this section, we study language equivalence between two SFSMs over the same concrete input domain $\DI$, using the same sets of input and output variables. Let $\Sigma_I\subseteq\DI$ be a set of concrete inputs. Two SFSMs $M=(Q,q_0,V,D,G,E,h)$ and $M'=(Q',q_0',V,D,G',E',h')$ are \emph{$\Sigma_I^*$-equivalent} if, for every input sequence $\alpha\in\Sigma_I^*$, both machines produce the same concrete output sequence. They are $\DI^*$-equivalent  when their languages coincide. We now give sufficient conditions under which equivalence on $\Sigma_I^*$ already implies equivalence on $\DI^*$.

\paragraph{Fault domains.}
The number of distinguishable reachable states of a machine is the number of language-equivalence classes of its reachable states. 
For DFSM testing with finite input alphabet $\Sigma_I$ and output alphabet $\Sigma_O$, a \emph{fault domain} ${\cal D}_{\mathrm{DFSM}}(m,\Sigma_I,\Sigma_O)$ consists of  all DFSMs over $\Sigma_I, \Sigma_O$ with at most $m$ distinguishable reachable states. Analogously, ${\cal D}_{\mathrm{SFSM}}(m,G,E)$ consists of all deterministic and completely specified SFSMs over the fixed variables and domains whose transitions use guards from $G$ and output assignments from $E$, and which have at most $m$ distinguishable reachable states. 
% We assume that there exists some SFSM $M'\in {\cal D}_{\mathrm{SFSM}}(m,G,E)$ representing the observable behaviour of the SUT ${\mathcal I}$, that is 
%  $\Lang({\mathcal I})=\Lang(M')$.

\paragraph{Test cases and test suites.}
We assume that the SUT can be reset reliably to its initial state before each test case. 

A \emph{test case} is a finite concrete input sequence $\alpha\in\DI^*$. Let $M$ be the reference model and let $M'$ represent the (initially unknown) observable behaviour of the SUT ${\mathcal I}$. The SUT \emph{passes} test case $\alpha$ if
\[
\out(q_0',\alpha)=\out(q_0,\alpha),
\]
where $q_0$ and $q_0'$ are the initial states of $M$ and $M'$, respectively.

A test suite $\TS$ is a finite set of test cases. The SUT passes $\TS$ if it passes every test case in this suite. 
In this paper, conformance means language equivalence.
A test suite $\TS$ is complete for ${\mathcal I}$ with respect to $M$ if 
\[{\mathcal I}\, \mathrm{passes}\, \TS \iff \Lang({\mathcal I})=\Lang(M).\]
It is complete for a fault domain ${\mathcal D}$ with respect to $M$ if it is complete for every SUT
 whose observable language coincides with that of some
model in ${\mathcal D}$. 

\begin{theorem}\label{th:main}
Let $M$ and $M'$ be deterministic, completely specified SFSMs over the same variables. Suppose that $\Sigma_I\subseteq \DI$ is representative for $(M,M')$.
Then, if $M$ and $M'$ are $\Sigma_I^*$-equivalent, they are  also $\DI^*$-equivalent. Hence $\Lang(M)=\Lang(M')$.
\end{theorem}

\begin{proof}
Assume that $M$ and $M'$ are $\Sigma_I^*$-equivalent. Let $\alpha=x_1\cdots x_k\in\DI^*$ be any input sequence. In $M$, the sequence $\alpha$ induces a uniquely defined symbolic guard sequence $g_1\cdots g_k$ and uniquely defined output assignments $e_1\cdots e_k$. In $M'$, it induces a  symbolic guard sequence $g_1'\cdots g_k'$ and output assignments $e_1'\cdots e_k'$ (also uniquely determined). Since $x_i\in g_i\cap g_i'$ for each $i$, all these intersections are non-empty.

By the assumed condition, for each $i$ we may choose $a_i\in g_i\cap g_i'\cap\Sigma_I$ such that $a_i$ separates $e_i$ and $e_i'$ whenever these assignments are distinguishable on $g_i\cap g_i'$. Let $\beta=a_1\cdots a_k\in\Sigma_I^*$. The sequence $\beta$ follows the same   guard sequences as $\alpha$ in both machines. Since the machines are $\Sigma_I^*$-equivalent, their output sequences on $\beta$ coincide. Thus $e_i(a_i)=e_i'(a_i)$ for every $i$. By the choice of $a_i$, the assignments $e_i$ and $e_i'$ cannot be distinguishable on $g_i\cap g_i'$. Since $x_i\in g_i\cap g_i'$, it follows that $e_i(x_i)=e_i'(x_i)$ for every $i$. Therefore the output sequences on $\alpha$ coincide. As $\alpha$ was arbitrary, the machines are $\DI^*$-equivalent.
\end{proof}

% \begin{remark}
% Under the assumptions of Theorem~\ref{th:main}, a reduced SFSM remains reduced after finite instantiation: (i) The guard-overlap requirement ensures that each reachable state under restriction to $\Sigma_I$ is still reachable. (2) If two states were equivalent in $M|_{\Sigma_I}$, Theorem~\ref{th:main} applied two copies of $M$ with these states as initial states then the same representative-set argument would imply their equivalence in $M$.
% \end{remark}

% \begin{corollary}\label{cor:mbmain1}
% Let $M=(Q,q_0,V,D,G,E,h)$ and $M'=(Q',q_0',V,D,G',E',h')$ be deterministic, completely specified SFSMs over the same variables. Suppose that $\Sigma_I\subseteq \DI$ is representative for $(G\cup G', E\cup E')$ or $(M,G\cup G', E\cup E')$.  Then, if 
% $M$ and $M'$ are $\Sigma_I^*$-equivalent, they are $\DI^*$-equivalent. 
% \end{corollary}
% \begin{proof}
%      Since $\Sigma_I$ is representative for $(G\cup G', E\cup E')$ or $(M,G\cup G', E\cup E')$, it is also representive for $(M,M')$, Theorem~\ref{th:main}  implies $\Lang(M)=\Lang(M')$. 
% \end{proof}

Theorem~\ref{th:main} allows complete DFSM test suites to be used for SFSMs. 
% If the chosen DFSM testing method requires a minimal reference model, we first minimise $M|_{\Sigma_I}$. The state count used by the testing method is the number of states of this minimal DFSM.
\begin{corollary}\label{cor:mbmain}
Let $M$ and $M'$ be deterministic, completely specified SFSMs over the same variables. Suppose that $\Sigma_I\subseteq \DI$ is representative for $(M,M')$. Define
$\Sigma_O=\{e(a)\mid a\in\Sigma_I,\ e\in E\cup E'\}$. Furthermore, suppose that $M'|_{\Sigma_I}\in {\cal D}_{\mathrm{DFSM}}(m,\Sigma_I,\Sigma_O)$. Let $\TS\subseteq \Sigma_I^*$ be a complete test suite for the DFSM fault domain
${\cal D}_{\mathrm{DFSM}}(m,\Sigma_I,\Sigma_O)$ with respect to  $M|_{\Sigma_I}$. Then 
\[M'\, \mathrm{passes}\, \TS \iff \Lang(M')=\Lang(M).\]
\end{corollary}
\begin{proof}
    Passing the complete test suite $\TS$ means that  $M$ and $M'$ are $\Sigma_I^*$-equivalent. Since $\Sigma_I$ is representative for $(M,M')$, Theorem~\ref{th:main}  implies $\Lang(M')=\Lang(M)$. 
    Conversely, if $\Lang(M')=\Lang(M)$, then the machines
agree on every input sequence, so $M'$ passes
the test suite.
\end{proof}

\begin{theorem}\label{th:mbmain}
Let $M$ be an SFSM with $n$ distinguishable reachable states, and let $n\le m$.  Let $G_h$ and $E_h$ be the sets of guards and output assignments occurring in $M$, respectively. Let $G$ and $E$ be finite sets of guards and output assignments, respectively, over the same variables and domains as $M$.
Let $\Sigma_I\subseteq D_I$
be a finite concrete input set representative for $(M,G,E)$. Define
$\Sigma_O=\{e(a)\mid a\in\Sigma_I,\ e\in E\cup E_h\}$.
Let $\TS\subseteq \Sigma_I^*$ be a complete DFSM test suite for the DFSM fault domain
${\cal D}_{\mathrm{DFSM}}(m,\Sigma_I,\Sigma_O)$, with respect to $M|_{\Sigma_I}$. 
Then
\[\forall M'\in {\cal D}_{\mathrm{SFSM}}(m,G,E)\colon M'\, \mathrm{passes}\, \TS \iff \Lang(M')=\Lang(M).\]
\end{theorem}

\begin{proof}
Let $M'\in {\cal D}_{\mathrm{SFSM}}(m,G,E)$. By Lemma~\ref{lemma:rep}, $\Sigma_I$ is representative for $(M,M')$. Restricting the input alphabet cannot split a language-equivalence class of reachable states, and every state reachable in the restriction is reachable in the original machine. Thus, the finite
instantiation ${M'}|_{\Sigma_I}$ has at most $m$ distinguishable reachable states. Its outputs belong to $\Sigma_O$, so ${M'}|_{\Sigma_I}\in {\cal D}_{\mathrm{DFSM}}(m,\Sigma_I,\Sigma_O)$.
Since $\TS$ is complete for ${\mathrm{DFSM}}(m,\Sigma_I,\Sigma_O)$ with respect to ${M}|_{\Sigma_I}$ and Theorem~\ref{th:main} holds for $(M,M')$, 
\[
 M'\, \mathrm{ passes }\, \TS 
\Longleftrightarrow \Lang({M'}|_{\Sigma_I})=\Lang(M|_{\Sigma_I})
\Longleftrightarrow 
\Lang(M')=\Lang(M).
\]
% \[
% \begin{aligned} M'\, \mathrm{ passes }\, \TS 
% &\Longleftrightarrow \Lang({M'}|_{\Sigma_I})=\Lang(M|_{\Sigma_I})\\
% &\Longleftrightarrow 
% \Lang(M')=\Lang(M).
% \end{aligned}
% \]

\end{proof}

% ======================================================================
\section{Model learning}\label{sec:learning}
We now consider how to learn an SFSM that represents the observable
behaviour of an SUT $\mathcal I$.

\subsection{Assumptions and overall strategy}\label{ssec:learning_assumption}
Assume that the SUT $\mathcal{I}$\footnote{We continue to use the term
\emph{System Under Test (SUT)}, although the term
\emph{System Under Learning (SUL)} is preferred by some authors.}  has an unknown deterministic and completely specified SFSM representation $M$ with $\Lang(M)=\Lang(\mathcal I)$. Static analysis provides a finite set $G'$ of branching
conditions over the input variables and a finite set $E$ containing all output assignments used by $M$. We assume that the input partition $\olg$ induced by $G'$ refines every guard occurring in $M$. Splitting these guards into classes from $\olg$ gives an SFSM
$\ol M$ over $(\olg, E)$ with the same observable behaviour.
Neither $M$ nor $\ol M$ needs to be known to the learner.

We also assume an upper bound $m$ on the number of
distinguishable reachable states of $M$. Such a bound may be obtained through
static analysis and abstract interpretation. The SUT must support a reliable reset to its initial state before each membership query and test execution. 

We first construct a finite representative input set $\Sigma_I$ for $(\olg, E)$. An $L^*$-style DFSM
learning algorithm then learns the finite instantiation $\ol M|_{\Sigma_I}$ by interacting with the SUT. Equivalence queries are answered by complete DFSM testing using the bound $m$. The final minimal DFSM hypothesis
$H_{\Sigma_I}$ is lifted to an SFSM hypothesis $H$. Theorem~\ref{th:learning} shows that $H$ and the SUT are language equivalent over the full input domain.

\subsection{Input equivalence classes}
The set $G'$ contains the branching conditions as they occur in the source code. For example, the statement
\begin{equation}\label{eq:codea}
\texttt{if}\ g_1\ \texttt{then}\ B_1\ \texttt{else if}\ g_2\ \texttt{then}\ B_2\ \texttt{else}\ B_3;
\end{equation}
contributes $g_1$ and $g_2$ to $G'$. Its three branches have guards $g_1, \neg g_1\wedge g_2, \neg g_1\wedge \neg g_2$. The negated conditions need not be collected separately.

 Two inputs are equivalent if they satisfy the same branching conditions:
\[
a\sim_{G'}b
\quad\Longleftrightarrow\quad
\forall g\in G'\colon
\bigl(a\in g \Longleftrightarrow b\in g\bigr).
\]
For each $U\subseteq G'$, define  
\begin{equation}\label{eq:classphi}
    \varphi_U = \bigwedge_{g\in U} g \wedge \bigwedge_{g\in G'\setminus U} \neg g.
\end{equation}
The non-empty sets $\varphi_U$ are the \emph{input equivalence classes}. Thus
\begin{equation}
\olg =\DI/_{\sim_{G'}}= \{ \varphi_U~|~U\subseteq G' \wedge \varphi_U\neq\varnothing \}
\end{equation}
is a finite partition of the input domain $\DI$. 

The refinement assumption from Section~\ref{ssec:learning_assumption} means that every guard occurring in $M$ is a union of classes from $\olg$:
\begin{equation}\label{eq:classref}
\forall g\in G_h\;\exists B_g\subseteq\olg\colon 
g=\bigcup_{\varphi\in B_g}\varphi.
\end{equation}

Define the refined machine by 
 \begin{equation}\label{eq:olM}
    \ol M =  (Q,q_0,V,D,\ol G,E,\ol h),\quad  \ol h = \{ (q_1,\varphi,e,q_2)\in Q\times \olg\times E\times Q~|~\exists g\in G_h, \varphi\subseteq g\wedge (q_1,g,e,q_2)\in h \}.
\end{equation}
Each transition of $M$ is split into one transition for each input class contained in its guard. The output assignment and the target state are unchanged.

\begin{lemma}\label{lemma:refine}
The SFSM $\ol M$ is deterministic and completely specified, and $\Lang(\ol M) = \Lang(M)$.
\end{lemma}
\begin{proof}
    Let $q\in Q$ and $a\in \DI$. Let $(q,g,e,q')\in h$ be the unique transition of $M $ enabled by $a$. Let $\varphi\in \olg$ be the unique input class containing $a$. By the refinement assumption~\eqref{eq:classref}, $\varphi\subseteq g$. Hence, by~\eqref{eq:olM}, $(q,\varphi, e, q')\in \ol h$ is enabled by $a$. 

    Any other transition of $\ol M$ enabled by $a$ in $q$ must use the same class $\varphi$ and arise from a transition of $M$ enabled by $a$. Since $M$ is deterministic, the output assignment and the target state must be $e$ and $q'$. This proves that $\ol M$ is deterministic.

    For every state and input, $M$ and $\ol M$ use the same output assignment and reach the same target state. Hence, by induction on the length of a given input sequence, it follows that $\Lang(M)=\Lang(\ol M)$.
\end{proof}
\subsection{Learning the finite instantiation}
% To learn the unknown  DFSM elements $(Q,q_0,h|_{\Sigma_I})$ of $M|_{\Sigma_I}$, any   $L^*$-style learning algorithm for DFSMs can be used, such as, for example, the $L^\#$ algorithm with its associated data structures supporting DFSM learning~\cite{DBLP:conf/tacas/VaandragerGRW22}. Regardless of the concrete algorithm, the learning process will come up with so-called  \emph{hypothesis DFSM models} $M|_{\Sigma_I}^{hyp}$ at certain stages of the process. To check whether a hypothesis $M|_{\Sigma_I}^{hyp}$ is really equivalent to $M|_{\Sigma_I}$, any complete testing method for DFSMs can be used to create a finite test suite with test cases from $\Sigma_I^*$ from hypothetical reference model $M|_{\Sigma_I}^{hyp}$ and run this against the SUT whose behaviour is equivalent to $M|_{\Sigma_I}$ when exercised with inputs from $\Sigma_I$ only. Complete DFSM testing methods use the upper state bound $m$ and the concrete number of (distinguishable) states in hypothesis 
% $M|_{\Sigma_I}^{hyp}$ to create test cases of appropriate length. If the generated test suite is passed by the SUT, language equivalence between $M|_{\Sigma_I}^{hyp}$ and $M|_{\Sigma_I}$ is guaranteed, provided that $m$ really is an upper bound for the number of control states in $M|_{\Sigma_I}$. 

Let $\Sigma_I$ be representative for $(\olg, E)$. Section~\ref{ssec:eq} gives an algorithm for constructing such a set. 
Any suitable $L^*$-style active learning algorithm
for DFSMs can be used to learn the SUT's behaviour over $\Sigma_I$. One example is the $L^\#$-algorithm~\cite{DBLP:conf/tacas/VaandragerGRW22}. Membership
queries are executed against the SUT using input from $\Sigma_I$. The learner constructs successive DFSM hypotheses, denoted by $H_{\Sigma_I}$.

% To learn a DFSM hypothesis $H_{\Sigma_I}$ whose behaviour coincides with
% that of the SUT $\mathcal I$ on input sequences from $\Sigma_I^*$, any suitable $L^*$-style active learning algorithm
% for DFSMs may be used under the assumption that reliable reset access to the initial state between membership
% queries and test executions. One example is the $L^\#$ algorithm together
% with its associated data structures for DFSM
% learning~\cite{DBLP:conf/tacas/VaandragerGRW22}. Regardless of the
% particular algorithm, the learner constructs successive hypothesis
% DFSMs, denoted by $H_{\Sigma_I}$, at different stages of the learning process.

\paragraph{Answering equivalence queries and the W-method optimisation.}
To determine whether a DFSM hypothesis $H_{\Sigma_I}$ is $\Sigma_I^*$-equivalent to
$\mathcal I$, a complete conformance-testing method for DFSMs, for example the W-method~\cite{vasilevskii1973,chow:wmethod}, can
be applied. Such a method takes the current hypothesis as the reference model and generates
a finite test suite $\TS\subseteq\Sigma_I^*$. 

The upper bound $m$ is also an upper bound on the number of states of a minimal DFSM equivalent to $\ol M|_{\Sigma_I}$: states equivalent over $\DI$ remain equivalent when inputs are restricted to $\Sigma_I$.

The test suite is executed against the SUT. A failed test supplies a counterexample for the learner. Under the state bound and the assumptions of the testing method, passing the test suite means
\[\Lang(H_{\Sigma_I})=\Lang(\ol M|_{\Sigma_I}).\]

\paragraph{Reducing the W-method traversal alphabet.}
Inputs in the same class $\varphi$ trigger the same symbolic transition in $\ol M$. We can therefore use one input per class in the traversal part of the W-method, while retaining the full alphabet $\Sigma_I$ for transition checks. 

The guard-overlap requirement ensures that for each input class $\varphi\in \olg$, $\varphi\cap \Sigma_I\neq \emptyset$. Select an input $a_\varphi\in \varphi\cap \Sigma_I$ for each $\varphi\in \olg$ and define a subset $\Sigma_I^{\mathrm{tr}}=\{a_\varphi\mid \varphi\in \olg\}\subseteq \Sigma_I$.

% \begin{proposition}   
% Let $H_{\Sigma_I}$ be a minimal
% DFSM hypothesis  with $n$ states, let $m\geq n$ be an upper bound on  the number of distinguishable reachable states of  $\ol M|_{\Sigma_I}$. Let $P$ be a state cover with $\varepsilon\in P$ and $W$ a non-empty characterisation set of $H_{\Sigma_I}$.
% Then, the test suite
% \[
% \TS_{\ol G}
% =
% P\cdot
% (\Sigma_I^{\mathrm{tr}})^{\leq m-n}
% \cdot
% \Sigma_I^{\leq 1}
% \cdot W
% \]
% is complete for  $\ol M|_{\Sigma_I}$.
% \end{proposition}
% Define the restricted DFSM fault domain
% \[ {\mathcal D}_{part}(m,\olg, E, \Sigma_I)=\{M'|_{\Sigma_I}\mid M'\in {\cal D}(m,G,E)\}.
% \]
\begin{proposition}   
Let $H_{\Sigma_I}$ be a minimal
DFSM hypothesis  with $n$ states, and $m\geq n$
be an upper bound on  the number 
of distinguishable reachable states of  $\ol M|_{\Sigma_I}$. 
Let $P$ be a state cover with $\varepsilon\in P$ and $W$ a non-empty characterisation set of $H_{\Sigma_I}$.
Define
\[
\TS_{\olg}
=
P\cdot
(\Sigma_I^{\mathrm{tr}})^{\leq m-n}
\cdot
\Sigma_I^{\leq 1}
\cdot W.
\]
Then, if $\ol M|_{\Sigma_I}$ passes $\TS_{\olg}$, $\Lang(H_{\Sigma_I})=\Lang (\ol M|_{\Sigma_I})$.
% is complete for ${\mathcal D}_{part}(m,\olg, E, \Sigma_I)$ with respect to $H_{\Sigma_I}$. In particular, passing the test suite establishes equivalence to $\ol M|_{\Sigma_I}$ under the learning assumption. 
\end{proposition}

\begin{proof}
Let $N$ be the reachable minimal quotient of $\ol M|_{\Sigma_I}$. Then $N$ has at most $m$ states and is language equivalent to $\ol M|_{\Sigma_I}$. Therefore, $\ol M|_{\Sigma_I}$ passes $\TS_{\olg }$ if and only if  $N$ passes $\TS_{\olg }$. Hence, it suffices to prove that $N$ passes $\TS_{\olg }$ implies $H_{\Sigma_I}$ and $N$ are language equivalent. 

Suppose $N$ passes $\TS_{\ol G}$. Since $P.W\subseteq \TS_{\olg}$ and $W$ distinguishes the states in $H_{\Sigma_I}$, the access sequences in $P$ reach at least $n$ distinct states in $N$ and $P.\Sigma_I^{\le m-n}$ is a state cover of $N$.  

Replace each input $a$ in the traversal part $\Sigma_I^{\le m-n}$ by $a_\varphi$, where $a\in \varphi$. Inputs in the same class reach the same successor in $\ol M|_{\Sigma_I}$, and this property is preserved in $N$. Hence $L=P.({\Sigma_I^{\mathrm{tr}}})^{\le m-n}$ is also a state cover of $N$. Since $P\subseteq L$, it is also a state cover of $H_{\Sigma_I}$.

Since $N$ passes $L.\Sigma_I^{\le 1}.W=\TS_{\olg}$ and $L$ is a state cover of both $H_{\Sigma_I}$ and $N$, Lemma~3.8 of ~\cite{DBLP:conf/tacas/KrugerJR24} implies that $H_{\Sigma_I}$ and $N$ are equivalent on $\Sigma_I^*$. 

% Conversely, suppose that $H_{\Sigma_I}$ and $N$ are equivalent on $\Sigma_I^*$. Then $N$ passes $\TS_{\olg}$.
\end{proof}
For a one-state hypothesis, we may take $W=\{\varepsilon\}$.

% This redueced test suite $\TS_{\olg}$ is complete for the restricted DFSM fault domain
% \[ {\mathcal D}_{part}(m,\olg, E, \Sigma_I)=\{M'|_{\Sigma_I}\mid M'\in {\cal D}(m,G,E)\}.
%  \]
We expect that this test-suite reduction can be adapted to variations of the W-method whose completeness arguments contain a traversal component used solely to reach possible additional SUT states. 

% However, the completeness of the resulting reduced test suite must be established for each such variant.
% This test-suite reduction can also be adapted to the H-method. Since the test suites generated by the W-, Wp-, and HSI-methods satisfy the sufficient conditions used in the H-method's completeness proof, the same arguments apply to these methods. Only the traversal component is restricted to $\Sigma_I^{\mathrm{tr}}$; transition checks and state distinguishing sequences still use the full finite input alphabet $\Sigma_I$. We omit the technical proof.   

% ......................................................................
\subsection{Lifting a learned DFSM to an SFSM}
Let $H_{\Sigma_I}=(S,s_0,\Sigma_I,\Sigma_O,h_{\Sigma_I})$ be a minimal DFSM hypothesis whose equivalence to $\ol M|_{\Sigma_I}$ has been established. We now construct an SFSM over $(\olg, E)$ from this hypothesis. 

Recall that, for $\varphi\in\olg$ and $e,e'\in E$,
\[
e\equiv_\varphi e'
\quad\Longleftrightarrow\quad
\forall a\in\varphi\colon e(a)=e'(a).
\]
For each $\varphi\in\olg$, choose a set $E_\varphi\subseteq E$ containing  exactly one
representative of every equivalence class of $E/{\equiv_\varphi}$. Define the SFSM $H=(S,s_0,V,D,\ol G,E,h^{hyp})$ by
\begin{equation}\label{eq:hhyp}
(s,\varphi,e,s')\in h^{hyp}
\quad\Longleftrightarrow\quad e\in E_\varphi \land 
\big( \forall a\in\Sigma_I\cap \varphi\centerdot (s,a,e(a),s') \in  h_{\Sigma_I}\big),
\end{equation}
where $s,s'\in S$, $\varphi\in \olg$ and $e\in E$.
 
% The transition relation $h^{{hyp}}$ is well-defined, since $\Sigma_I$ is representative for $(\ol G,E)$. In particular, every input $a\in\Sigma_I$ belongs to exactly one equivalence class $\varphi\in\ol G$, and $\Sigma_I\cap\varphi$ contains sufficiently many inputs to distinguish output expressions that are inequivalent on $\varphi$. Moreover, since $\ol G$ refines the guards of the unknown SFSM, all inputs in a class $\varphi$ induce the same successor state and output behavior. Consequently, the samples in $\Sigma_I\cap\varphi$ uniquely determine the successor state $q'$ and the output expression $e$, up to equivalence on $\varphi$, through the requirement $(q,a,e(a),q')\in h_{\Sigma_I}
% $ for every $a\in\Sigma_I\cap\varphi.$

% At any fixed state of $\ol M$, all inputs belonging to the same input
% equivalence class $\varphi$ enable the same symbolic transition and hence
% lead to the same successor state and are processed by the same output
% expression. The concrete output values may differ because
% they are determined by $e(a)$. 

The following lemma shows that $H$ is deterministic and completely specified
and preserves the behaviour of the DFSM hypothesis $H_{\Sigma_I}$ over $\Sigma_I$.

% Since any input set that is representative for $(\olg,E)$ is also representative for $(M,\olg,E)$, it suffices to prove the following lemma under the assumption that $\Sigma_I$ 
% is representative for $(M,\olg,E)$. The result then applies, in particular, when $\Sigma_I$
% is representative for $(\olg,E)$.

\begin{lemma}\label{lemma:lifting-properties}
Let $\ol M=(Q,q_0, V, D, \ol G, E, \ol h)$ be a deterministic and completely specified SFSM whose guard set $\olg$ is an input partition. Let 
$\Sigma_I\subseteq\DI$ be representative for $(\olg,E)$, and let $H_{\Sigma_I}
=(S,s_0,\Sigma_I,\Sigma_O,h_{\Sigma_I})$ be a minimal, deterministic, and completely specified
DFSM satisfying
$\Lang(H_{\Sigma_I})=\Lang(\ol M|_{\Sigma_I})$.
Then the SFSM $H=(S,s_0, V,D,\olg, E,h^{hyp})$ defined by~\eqref{eq:hhyp} is deterministic and completely specified, and
$\Lang(H|_{\Sigma_I})=\Lang(\ol M|_{\Sigma_I})$.
\end{lemma}

\begin{proof}
Let $s\in S$ and $\varphi\in\olg$ be arbitrary. By the guard-overlap condition, $\varphi\cap \Sigma_I\neq \emptyset$. 
Since $H_{\Sigma_I}$ is minimal, there exists an input sequence $\alpha\in\Sigma_I^*$ that reaches $s$. Let $q$ be the state reached by $\ol M$ after $\alpha$. Since $\ol M$ is deterministic and completely specified, there exists a unique
successor state $q'$ and an output expression $e\in E$ such that $(q,\varphi,e,q')\in \ol h$. Choose  $\widehat e\in E_\varphi$ with $\widehat e\equiv_\varphi e$.

For any $a\in \varphi\cap\Sigma_I$, let $(s,a,o_a,s_a)\in h_{\Sigma_I}$
be the corresponding transition of $H_{\Sigma_I}$. The language equivalence
$\Lang(H_{\Sigma_I})=\Lang(\ol M|_{\Sigma_I})$
implies that $H_{\Sigma_I}$ and $\ol M$ produce the same output on input $a$ after executing $\alpha$. Hence, $o_a=e(a)=\widehat e(a)$.

We claim that the successor state $s_a$ is independent of the choice of
$a\in \varphi\cap\Sigma_I$. Let $a,b\in \varphi\cap\Sigma_I$. In $\ol M$, both $a$ and $b$
take the state $q$ to the same successor state $q'$. Therefore, the continuation behaviours of $\ol M$ after
$\alpha a$ and $\alpha b$ coincide. By the language equivalence of $H_{\Sigma_I}$ and
$\ol M|_{\Sigma_I}$, the states $s_a$ and $s_b$ have the same
input-output behaviour over $\Sigma_I$. Since $H_{\Sigma_I}$ is minimal, $s_a=s_b$.
Let $s'$ denote this common successor state. We then have
$\forall a\in \varphi\cap\Sigma_I \colon (s,a,\widehat e(a),s')\in h_{\Sigma_I}$.
By Equation~\eqref{eq:hhyp}, $(s,\varphi,\widehat e,s')\in h^{hyp}$.
Since $\olg$ partitions $\DI$, $H$ is completely specified.

To prove determinism, consider two lifted transitions $(s,\varphi_i,e_i,s_i')\in h^{hyp}$, $i=1,2$, from $s$ with overlapping guards, i.e., $\varphi_1\cap \varphi_2\neq \emptyset$. Since $\olg$ partitions $\DI$, $\varphi_1=\varphi_2$. Let $\varphi=\varphi_1=\varphi_2$.
Since $\varphi\cap\Sigma_I\neq\emptyset$ and $H_{\Sigma_I}$ is deterministic,
we obtain $s_1'=s_2'$ and $e_1(a)=e_2(a)$ for every $a\in \varphi\cap\Sigma_I$. The output-separation requirement for $\Sigma_I$ implies $e_1\equiv_\varphi e_2$; otherwise there exists $a\in \Sigma_I\cap \varphi$ satisfying $e_1(a)\neq e_2(a)$. Since $e_1,e_2\in E_{\varphi}$, it follows that $e_1=e_2$, and $H$ is deterministic.  

By~\eqref{eq:hhyp}, every lifted transition $(s,\varphi,e,s')\in h^{{hyp}}$ restricts to transitions 
$(s,a,e(a),s')\in h_{\Sigma_I}$ of $H_{\Sigma_I}$, for every $a\in \varphi\cap \Sigma_I$.
Since $H$ is completely specified, any transition of $H_{\Sigma_I}$ is obtained in this way. 
Hence, $H|_{\Sigma_I}$ and $H_{\Sigma_I}$ have the same transitions and initial state. Therefore, $\Lang(H|_{\Sigma_I})=\Lang(H_{\Sigma_I})=\Lang(\ol M|_{\Sigma_I})$.

\end{proof}

% \begin{theorem}\label{th:learning}
% Let $\ol M$ be an unknown deterministic and completely specified SFSM over $(\ol G,E)$, and let $\Sigma_I$ be representative for $(\ol G,E)$. Suppose that a DFSM learning algorithm operating over the input alphabet $\Sigma_I$ terminates with a hypothesis $H_{\Sigma_I}$ such that  $\Lang(H_{\Sigma_I})=\Lang(\ol M|_{\Sigma_I})$.
% Then the lifted hypothesis $H$, constructed from $H_{\Sigma_I}$ according to Equation~\eqref{eq:hhyp}, satisfies
% $\Lang(H)=\Lang(\ol M)$.

% \end{theorem}

\begin{theorem}\label{th:learning}
Let \(\mathcal I\) be an SUT with a deterministic and completely specified SFSM representation $ M=(Q,q_0,V,D,G,E,h)$ such that $\Lang(\mathcal I)=\Lang(M)$.

Let $\olg$ be a finite input partition that refines every guard occurring in $M$, let $\ol M$ be the refinement defined by~\eqref{eq:olM} and let $\Sigma_I$ be representative for $(\olg, E)$.

Suppose that a DFSM learning algorithm over $\Sigma_I$ returns a
minimal hypothesis $H_{\Sigma_I}$ satisfying
\[
\Lang(H_{\Sigma_I})
=
\Lang(M|_{\Sigma_I}).
\]
Then the SFSM $H$ defined by~\eqref{eq:hhyp} satisfies
\[
\Lang(H)
=
\Lang(\mathcal I).
\]
\end{theorem}

\begin{proof}
By Lemma~\ref{lemma:refine},
\(\Lang(\ol M)=\Lang(M)\), and therefore $\Lang(\ol M|_{\Sigma_I})=\Lang( M|_{\Sigma_I})$.
Since $\Lang(H_{\Sigma_I})=\Lang(M|_{\Sigma_I})=\Lang(\ol M|_{\Sigma_I})$, Lemma~\ref{lemma:lifting-properties} implies 
that $H$ and $\ol M$ are equivalent over $\Sigma_I$. Since $\Sigma_I$ is representative for $(\ol G,E)$, Theorem~\ref{th:main} implies that they are equivalent over $\DI$.
Consequently, $\Lang(H)=\Lang(\ol M)=\Lang(M)=\Lang(\mathcal I)$.
\end{proof}

% ......................................................................
\subsection{Coarsening the learned SFSM}
The lifted SFSM $H$ has transitions for each input class $\varphi\in\olg$ at each state. For model checking, it can be useful to reduce the number of transitions. We merge the transitions with the same source state, target state and output assignment by taking the union of their guards. For any $s,s'\in S$ and $e\in E$, define:
\[ U_{(s,e,s')}=\{\varphi\in \olg\mid (s,\varphi,e,s')\in h^{hyp}\},\qquad  g_{(s,e,s')}=\bigcup_{\varphi\in U_{(s,e,s')}}\varphi.\]
%\[g_{(s,e,s')}=\bigcup_{\varphi\in U_{(s,e,s')}}\varphi.\]
Let
\[\hat G=\{ g_{(s,e,s')}\mid (s,e,s')\in S\times E\times S, g_{(s,e,s')}\neq \emptyset\}, \]
and 
\[\hat h =\{(s,g_{(s,e,s')},e,s') \mid s,s'\in S, e\in E, g_{(s,e,s')}\in \hat G\}.\] 
 The coarsened SFSM is $\hat H=(S,s_0,V,D,\hat{G},E,\hat h)$. 

 At each state, the merged guards remain disjoint and cover $\DI$. Thus, $\hat H$ is deterministic and completely specified. 
Refining its guards with respect to $\olg$ reconstructs $H$. 
Hence, by Lemma~\ref{lemma:refine}, $\Lang(\hat H) = \Lang(H)$.

\subsection{A learning and coarsening example}
\label{ssec:learning-example}

Consider the SFSM $M$ from Example~\ref{ex:sfsm} and suppose
that it represents the behaviour of the SUT $ \mathcal I$. Let $m=4$,
$G'=\{g_1,g_3,g_5\}$, and $E=\{e_1,\dots,e_5\}$, with
guards and output assignments as defined in that example.
The input equivalence partition induced by $G'$ is
$\olg=\{\varphi_1,\dots,\varphi_8\}$, where
\[
\begin{aligned}
&\varphi_1=g_1\cap g_3\cap g_5,
&&\varphi_2=g_1\cap g_3\cap\neg g_5,\\
&\varphi_3=g_1\cap\neg g_3\cap g_5,
&&\varphi_4=g_1\cap\neg g_3\cap\neg g_5,\\
&\varphi_5=\neg g_1\cap g_3\cap g_5,
&&\varphi_6=\neg g_1\cap g_3\cap\neg g_5,\\
&\varphi_7=\neg g_1\cap\neg g_3\cap g_5,
&&\varphi_8=\neg g_1\cap\neg g_3\cap\neg g_5.
\end{aligned}
\]
All eight classes are non-empty. Choose $a_i\in\varphi_i$ as
\[
(a_1,\dots,a_8)=(-2,-10,-3,-11,8,14,7,11)
\]
and let $\Sigma_I=\{a_1,\dots,a_8\}$.
At each selected input $a_i$, the five output assignments
have pairwise distinct values, as shown in Table~\ref{tab:learning-outputs}. Thus, $\Sigma_I$ meets every
class and fulfils the output-separation condition for
$(\olg,E)$.

Suppose that the DFSM learner has obtained the hypothesis
$H_{\Sigma_I}$ in Table~\ref{tab:learning-hypothesis} and established its equivalence to the
SUT over $\Sigma_I$ by complete testing with state bound
$m=4$. The initial state is $s_0$; each table entry gives
the output and successor state. The states $s_0,s_1,s_2,s_3$ are reached by $\varepsilon$, $a_5$, $a_5a_5$, and $a_5a_7$, respectively.
Input $a_5=8$ produces the distinct outputs $16,8,64,0$
at these states. Hence, the hypothesis is reachable and
minimal.

\begin{table}[htbp]
    \centering
    \begin{tabular*}{.95\linewidth}
        {@{\extracolsep{\fill}}|c|c|c|c|c|c|c|c|c|}
        \hline
       $H_{\Sigma_I}$ & $a_1$ & $a_2$ & $a_3$ & $a_4$
        & $a_5$ & $a_6$ & $a_7$ & $a_8$ \\ \hline
        $s_0$
        & $(0,s_0)$ & $(0,s_0)$ & $(0,s_0)$ & $(0,s_0)$
        & $(16,s_1)$ & $(28,s_1)$ & $(14,s_1)$ & $(22,s_1)$
        \\ \hline
        $s_1$
        & $(-2,s_2)$ & $(-10,s_2)$ & $(-2,s_3)$ & $(-10,s_3)$
        & $(8,s_2)$ & $(14,s_2)$ & $(8,s_3)$ & $(12,s_3)$
        \\ \hline
        $s_2$
        & $(4,s_2)$ & $(100,s_2)$ & $(-6,s_3)$ & $(-22,s_3)$
        & $(64,s_2)$ & $(196,s_2)$ & $(14,s_3)$ & $(22,s_3)$
        \\ \hline
        $s_3$
        & $(0,s_2)$ & $(0,s_3)$ & $(0,s_2)$ & $(0,s_3)$
        & $(0,s_2)$ & $(0,s_3)$ & $(0,s_2)$ & $(0,s_3)$
        \\ \hline
    \end{tabular*}
    \caption{The hypothesis $H_{\Sigma_I}$.
    Each entry gives the output and successor state.}
    \label{tab:learning-hypothesis}
\end{table}

The values of the output assignments at the selected inputs are given in Table~\ref{tab:learning-outputs}. 
\begin{table}[htbp]
    \centering
    \begin{tabular*}{.85\linewidth}
        {@{\extracolsep{\fill}}|c|c|c|c|c|c|c|c|c|}
        \hline
        & $a_1$ & $a_2$ & $a_3$ & $a_4$
        & $a_5$ & $a_6$ & $a_7$ & $a_8$ \\ \hline
        $e_1$ & $0$ & $0$ & $0$ & $0$
        & $0$ & $0$ & $0$ & $0$ \\ \hline
        $e_2$ & $-2$ & $-10$ & $-3$ & $-11$
        & $8$ & $14$ & $7$ & $11$ \\ \hline
        $e_3$ & $4$ & $100$ & $9$ & $121$
        & $64$ & $196$ & $49$ & $121$ \\ \hline
        $e_4$ & $-4$ & $-20$ & $-6$ & $-22$
        & $16$ & $28$ & $14$ & $22$ \\ \hline
        $e_5$ & $-1$ & $-9$ & $-2$ & $-10$
        & $9$ & $15$ & $8$ & $12$ \\ \hline
    \end{tabular*}
    \caption{Output-assignment values at the selected inputs. Each entry in row $e_j$ and column $a_i$ is $e_j(a_i)$.
    The five assignments have pairwise distinct values
    in each column.}
    \label{tab:learning-outputs}
\end{table}

We now lift $H_{\Sigma_I}$ to an SFSM $H$ over $\olg$. Since each column of Table~\ref{tab:learning-outputs} contains distinct values, the symbolic
transitions can be read directly from this table: replace
each label $a_i/b$ in $H_{\Sigma_I}$ by $\varphi_i/e_j$,
where $e_j$ is the unique assignment satisfying $e_j(a_i)=b$.
The source and successor states remain unchanged.
% Each concrete transition
% $s\xrightarrow{a_i/b}s'$ determines a unique assignment
% $e\in E$ with $e_j(a_i)=b$ and yields the symbolic transition
% $s\xrightarrow{\varphi_i/e}s'$.
For example,
\[
s_1\xrightarrow{a_3/(-2)}s_3
\quad\text{yields}\quad
s_1\xrightarrow{\varphi_3/e_5}s_3,
\]
since $a_3=-3$ and $e_5(-3)=-2$.
The lifted SFSM has eight outgoing transitions at each
state, hence $32$ transitions in total.

\begin{table}[htbp]
    \centering
    \begin{tabular*}{.85\linewidth}
        {@{\extracolsep{\fill}}|c|c|c|c|c|c|c|c|c|}
        \hline
       $H$ & $\varphi_1$ & $\varphi_2$ & $\varphi_3$ & $\varphi_4$
        & $\varphi_5$ & $\varphi_6$ & $\varphi_7$ & $\varphi_8$
        \\ \hline
        $s_0$
        & $(e_1,s_0)$ & $(e_1,s_0)$ & $(e_1,s_0)$ & $(e_1,s_0)$
        & $(e_4,s_1)$ & $(e_4,s_1)$ & $(e_4,s_1)$ & $(e_4,s_1)$
        \\ \hline
        $s_1$
        & $(e_2,s_2)$ & $(e_2,s_2)$ & $(e_5,s_3)$ & $(e_5,s_3)$
        & $(e_2,s_2)$ & $(e_2,s_2)$ & $(e_5,s_3)$ & $(e_5,s_3)$
        \\ \hline
        $s_2$
        & $(e_3,s_2)$ & $(e_3,s_2)$ & $(e_4,s_3)$ & $(e_4,s_3)$
        & $(e_3,s_2)$ & $(e_3,s_2)$ & $(e_4,s_3)$ & $(e_4,s_3)$
        \\ \hline
        $s_3$
        & $(e_1,s_2)$ & $(e_1,s_3)$ & $(e_1,s_2)$ & $(e_1,s_3)$
        & $(e_1,s_2)$ & $(e_1,s_3)$ & $(e_1,s_2)$ & $(e_1,s_3)$
        \\ \hline
    \end{tabular*}
    \caption{The lifted SFSM $H$.
    An entry $(e_j,s')$ in row $s$ and column $\varphi_i$
    denotes the transition $s\xrightarrow{\varphi_i/e_j}s'$.}
    \label{tab:learning-lifted}
\end{table}
Coarsening merges transitions with the same source state,
successor state, and output assignment.
Table~\ref{tab:learning-coarsening} lists the groups and their merged guards; an index
set $J$ denotes the classes $\{\varphi_i\mid i\in J\}$.
\begin{table}[htbp]
    \centering
    \begin{tabular*}{.85\linewidth}
        {@{\extracolsep{\fill}}|c|c|c|c|c|}
        \hline
        Source & $J$ & Assignment & Successor
        & $\bigcup_{i\in J}\varphi_i$
        \\ \hline
        $s_0$ & $\{1,2,3,4\}$ & $e_1$ & $s_0$ & $g_1$
        \\ \hline
        $s_0$ & $\{5,6,7,8\}$ & $e_4$ & $s_1$ & $g_2$
        \\ \hline
        $s_1$ & $\{1,2,5,6\}$ & $e_2$ & $s_2$ & $g_3$
        \\ \hline
        $s_1$ & $\{3,4,7,8\}$ & $e_5$ & $s_3$ & $g_4$
        \\ \hline
        $s_2$ & $\{1,2,5,6\}$ & $e_3$ & $s_2$ & $g_3$
        \\ \hline
        $s_2$ & $\{3,4,7,8\}$ & $e_4$ & $s_3$ & $g_4$
        \\ \hline
        $s_3$ & $\{1,3,5,7\}$ & $e_1$ & $s_2$ & $g_5$
        \\ \hline
        $s_3$ & $\{2,4,6,8\}$ & $e_1$ & $s_3$ & $g_6$
        \\ \hline
    \end{tabular*}
    \caption{Coarsening the lifted SFSM.
    Each row merges the transitions with guards
    $\varphi_i$, $i\in J$, into one transition.}
    \label{tab:learning-coarsening}
\end{table}

Thus, the coarsened SFSM $\hat H$ has eight transitions
and coincides with $M$ up to the renaming $s_i\mapsto q_i$.
In particular,
\[
\Lang(\hat H)=\Lang(H)=\Lang(M)=\Lang(\mathcal I).
\]

% \subsection{A learning and coarsening example}
% Consider the SFSM $M$ given in Example~\ref{ex:sfsm}. Suppose that $M$ is a SFSM representative of the SUT. 
% Let $m=4$, $G'=g_1,g_3,g_5\}$ and $E=\{e_1,\dots, e_5\}$, where $g_i, e_j$ are defined in Example~\ref{ex:sfsm}.  We first calculate the input equivalence partition $\olg $ form $G'$ and obtain $8$ equivalence classes:
% \begin{aligned}
%   & \varphi_1=g_1\cap g_3\cap g_5, \qquad  \varphi_2=g_1\cap g_3\cap \neg g_5,\\
%    & \varphi_3=g_1\cap \neg g_3\cap g_5, \qquad  \varphi_4=g_1\cap \neg g_3\cap \neg g_5,\\
%     & \varphi_5=\neg g_1\cap g_3\cap g_5, \qquad  \varphi_6=\neg g_1\cap g_3\cap  \neg g_5,\\
%      & \varphi_7=\neg g_1\cap \neg g_3\cap g_5, \qquad  \varphi_8=\neg g_1\cap \neg g_3\cap \neg g_5,\\
% \end{aligned}
% Select $a_i\in \varphi_i$ by $a_1=-2,a_2=-10, a_3=-3, a_4=-11, a_5=8, a_6=14, a_7=7, a_8=11$. Let
%  $\Sigma_I=\{a_i\mid i=1,\dots 8\}$. It fulfils the output-separation condition for $(\olg, E)$.  Suppose the DFSM learner has obtain the following hypothesis $H_{\Sigma_I}$ and establish it equivalence by complete testing.
%=========================================================================
%\input{section6_construction}

\section{Constructing finite representative input sets}\label{sec:algorithms}
Definition~\ref{def:repGE} requires witnesses for non-empty guard overlaps and separating inputs for output assignments that differ on those overlaps. We give a direct construction and prove that it produces a finite representative set. Appendix~\ref{app:algorithms} gives further constructions and procedures for reducing the number of representatives.

\subsection{Solver assumptions}\label{ssec:construction-assumptions}
Let $G=\{g_1,\ldots,g_r\}$ and $E=\{e_1,\ldots,e_t\}$ be finite sets of guards and total output assignments. We assume that the formulas $g_i(x)\wedge g_j(x)$ and $g_i(x)\wedge g_j(x)\wedge(e_p(x)\neq e_q(x))$ can be expressed in a decidable theory supported by an SMT solver. Every call to $\smta$ must terminate with a correct satisfiability result. For a satisfiable formula, it must also return a witness $a\in\DI$. In the pseudocode, $\mathsf{sat}$ is a Boolean flag; the witness is undefined when this flag is false. These assumptions concern the effective construction of representatives; the equivalence results above only require a representative set satisfying the stated conditions.

\subsection{Construction and correctness}
Algorithm~\ref{alg:SigmaI-basic} checks each guard pair, including pairs with $i=j$. For every non-empty overlap, it obtains a witness for each pair of output assignments that can be distinguished there. If no such pair exists, it retains the initial guard-overlap witness. Thus, a guard-overlap witness is added only when no separating witness has already been added for that overlap.

\begin{algorithm}[H]
\caption{Construct a representative input set for $(G,E)$}
\label{alg:SigmaI-basic}
\begin{algorithmic}[1]
\REQUIRE $G=\{g_1,\ldots,g_r\}$ and $E=\{e_1,\ldots,e_t\}$.
\ENSURE A finite representative set $\Sigma_I\subseteq\DI$ for $(G,E)$.
\STATE $\Sigma_I\gets\emptyset$
\FOR{$1\le i\le j\le r$}
  \STATE $(\mathsf{sat},a_{ij})\gets\smta(g_i\wedge g_j)$
  \IF{$\mathsf{sat}$}
    \STATE $\mathsf{sep}\gets\textbf{false}$
    \FOR{$1\le p<q\le t$}
      \STATE $(\mathsf{sat},a)\gets\smta(g_i\wedge g_j\wedge(e_p(x)\neq e_q(x)))$
      \IF{$\mathsf{sat}$}
        \STATE $\Sigma_I\gets\Sigma_I\cup\{a\}$; $\mathsf{sep}\gets\textbf{true}$
      \ENDIF
    \ENDFOR
    \IF{$\neg\mathsf{sep}$}
      \STATE $\Sigma_I\gets\Sigma_I\cup\{a_{ij}\}$
    \ENDIF
  \ENDIF
\ENDFOR
\STATE \textbf{return} $\Sigma_I$
\end{algorithmic}
\end{algorithm}

\begin{proposition}\label{prop:construction}
Under the assumptions in Section~\ref{ssec:construction-assumptions}, Algorithm~\ref{alg:SigmaI-basic} terminates and returns a finite representative input set for $(G,E)$. It makes at most
\[
  \frac{r(r+1)}{2}\left(1+\frac{t(t-1)}{2}\right)
\]
SMT solver calls.
\end{proposition}
\begin{proof}
All loops range over finite sets, and each solver call terminates. Hence the algorithm terminates and returns a finite set.
Consider a non-empty overlap $g_i\cap g_j$. The first solver call returns a witness $a_{ij}$. If some output assignments differ on this overlap, an inner solver call adds a separating witness, which also belongs to the overlap. Otherwise, the algorithm adds $a_{ij}$. Thus every non-empty guard overlap is represented.
For each pair of output assignments that differs on an overlap, the corresponding inner formula is satisfiable, and its returned witness is added to $\Sigma_I$. Since guard intersection and output inequality are symmetric, considering $i\le j$ and $p<q$ covers all required pairs. Both conditions of Definition~\ref{def:repGE} therefore hold.
There are $r(r+1)/2$ guard pairs. Each requires one overlap test and at most $t(t-1)/2$ separation tests, which gives the stated bound.
\end{proof}

The bound counts solver calls, not running time; the cost of each call depends on the underlying theory and the generated formula. The construction does not guarantee a representative set of minimum size.

\subsection{Use in testing and learning}
For learning, the construction can be applied to $(\olg,E)$, where $\olg$ is the known input partition. For testing against a reference model with guards $G_h$ and output assignments $E_h$, it can be applied to $(G_h\cup G,E_h\cup E)$. The resulting general representative set also satisfies the model-specific requirements.

Appendix~\ref{app:algorithms} gives constructions tailored to input equivalence classes, a fixed reference model, a pair of models, and constant output assignments. It also gives MaxSMT variants and a greedy set-cover procedure. The latter can reduce a candidate set while preserving all requirements for representativeness, but does not guarantee a minimum-size set. These procedures are given as theoretical constructions; their implementation and experimental evaluation are left for future work.

% % =================================================
%======================
\section{Related work}\label{sec:related}

\paragraph{Model-based testing for symbolic finite-state machines.}
Petrenko and Sim{\~a}o~\cite{PetrenkoSimao2015} extended checking
experiments to FSMs with symbolic inputs. Petrenko~\cite{DBLP:journals/sosym/Petrenko19} later considered
machines with symbolic inputs and outputs and distinguished assignment/output faults
from a more general class of transition faults.
In the transition-fault domain, implementations
may use arbitrary guards, while their output assignments must belong to a 
fixed finite set $E$. The  number of implementation states is also bounded by
the number of states of the reference model. For an infinite input domain, completeness
of a concrete test instantiation further requires the
$\Omega^d$-converter condition~\cite[Theorem 6]{DBLP:journals/sosym/Petrenko19}. This condition requires each reference-defined symbolic input sequence of length $d$ to produce a single symbolic output sequence in the implementation. Here, $d$ is chosen so that every state of the reference model is reachable by a defined symbolic input sequence shorter than $d$.
%=========================
%============================================

The fault domain considered in this paper is defined differently. We assume that finite
sets $G$ and $E$ are given such that the observable SUT behaviour has an SFSM representation whose transitions use guards from $G$ and output assignments from $E$. This assumption concerns the existence of such a representation; the SUT itself need not  use the
guards from $G$. The upper state bound may exceed the number of states of the reference model. Under these assumptions, no
$\Omega^d$-converter condition is needed.

Our testing method is based on a finite representative input set $\Sigma_I$. Under the guard-overlap and output-separation conditions, language equivalence of the finite instantiations over $\Sigma_I$ implies language equivalence of the SFSMs over $\DI$.

Taromirad and
Mousavi~\cite{TaromiradMousavi2017} study  grey-box conformance  testing for symbolic reactive state machines. Their test generation is based on a transition composition of the specification and an abstraction of the particular implementation. 
Our model-specific construction requires a reference model and finite sets of admissible guards and output assignments; it does not require a separate abstraction of the implementation. It first restricts the concrete input domain to a finite representative input set. The resulting finite instantiations can then be tested by any complete DFSM testing method whose assumptions are satisfied.

\paragraph{Active learning of symbolic state machines.}
Classical active automata learning infers finite-state models through
membership and equivalence queries. Several algorithms have been developed for learning Mealy machines~\cite{Niese03, ShahbazG09,IsbernerHS14,DBLP:conf/tacas/VaandragerGRW22}. 
These algorithms assume a finite input alphabet.
Other approaches address large or infinite input domains by refining
finite alphabet abstractions or by learning symbolic transition
predicates during the learning process
\cite{howarEtAl2011,aartsEtAl2012,drewsDAntoni2017,
ArgyrosDAntoni2018}.

Irie et al.~\cite{IrieWagaSuenaga2026} propose the \(\Lambda_M^*\) algorithm for learning deterministic and complete symbolic Mealy automata over potentially infinite input alphabets. Their algorithm maintains a finite set of essential input characters, refines this set during learning, and constructs a finite Mealy hypothesis whose transitions are subsequently generalised using symbolic predicates.
The use of finite concrete inputs followed by symbolic generalisation is closely related to our approach. An important difference concerns the output model. 
In their model, each symbolic transition produces a constant output. Our SFSMs use a known finite set of input-dependent output assignments and allow infinite concrete output domains.

The method presented here starts from information obtained before the learning process.
Static analysis provides a finite set $G'$ of branching conditions and a finite set $E$
of output assignments. The input equivalence
classes $\olg$ induced by $G'$ are assumed to refine the guards of an SFSM representation of the observable SUT behaviour. From these classes $\olg$ and output
assignments $E$, we construct a finite representative input alphabet
$\Sigma_I$. A finite-alphabet DFSM learning algorithm,
such as the $L^\#$-algorithm
\cite{DBLP:conf/tacas/VaandragerGRW22}, can then learn the
finite instantiation over $\Sigma_I$.

Equivalence queries are answered by complete DFSM conformance
testing, using an upper bound on the number of
distinguishable reachable states. Once a DFSM hypothesis equivalent to the
 SUT over $\Sigma_I$ has been learned, it is lifted
to an SFSM hypothesis. Theorem~\ref{th:learning} then gives language equivalence over the full input domain.
The contribution is therefore not a new active learning algorithm, but the reduction to finite-alphabet learning and the subsequent lifting result.
Although the learned finite instantiation has a finite output alphabet,
the concrete output domain of the SUT may be infinite.

%---------------------------------------

\section{Conclusion}\label{sec:conc}
Finite representative input sets reduce 
testing and learning of SFSMs to testing and learning of DFSMs over finite alphabets. A representative set contains witnesses for the relevant guard overlaps and separating inputs for the relevant pairs of output assignments. Under these conditions, equivalence of the finite instantiations implies language equivalence over the full concrete input domain.

For model-based testing, this result
allows complete DFSM testing methods to be
applied to finite instantiations of an SFSM reference model and of machines in the associated fault domain. 
For model learning, a finite instantiation can first be learned by a standard DFSM learning algorithm and then lifted to a language-equivalent SFSM. 
We also present SMT- and MaxSMT-based algorithms for constructing representative input
sets. A greedy set-cover procedure can remove candidates that are not needed.
The constructions cover general, input-equivalence-class, model-specific, and pair-specific representative sets.
 
\paragraph{Limitations.}
The results apply to deterministic and completely
specified SFSMs whose output assignments depend only on the current
input. For model-based testing, finite sets of admissible guards and output assignments must be known, together with an upper bound on the number of
distinguishable reachable states of an SFSM representation of the SUT. 

For model learning, static analysis must provide a 
finite set $G'$ of branching conditions whose induced partition refines the guards of a suitable SFSM representation and may contain up to $2^{|G'|}$ classes. It must also provide a finite set containing all output assignments used by that representation.
Answering equivalence queries by complete DFSM testing again requires an upper state bound. 
Static analysis and abstract interpretation may provide this information for software systems, but their precision and cost affect the applicability of the method.

\paragraph{Future work.}
An experimental evaluation should compare the SMT- and MaxSMT-based constructions. 
Relevant measures include solver runtime, the size of the resulting
representative input alphabet, and the  number and length of test cases and learning queries. 

The model class can be extended in several ways. Relevant cases include
partial and nondeterministic SFSMs, as well as output assignments that depend on
internal data, stored register values, previous inputs, or previous
outputs. It is also useful to consider systems for which the relevant guards and output
assignments cannot be determined completely before testing or learning
starts. 

A different learning scheme could operate directly on the finite alphabet $\olg$ of input equivalence classes, rather than first constructing and
learning a finite concrete instantiation. Such a scheme would require a specialised teacher that translates symbolic queries into concrete interaction with the SUT and relates the observable outputs to the available output assignments. 

Learning SFSM hypotheses could also be combined with symbolic model checking  
for black-box checking of systems
with large or infinite data domains. Further work should examine representative input sets constructed for particular properties or classes of properties. 

\section{Declaration of AI use}
During the preparation of this work, the authors used ChatGPT Pro
to improve the clarity, readability, and 
presentation of the manuscript. The tool was also used to generate TikZ code for Figures 1--3, assist in revising the algorithms and search for potentially relevant literature.
The scientific content,
definitions, formal results, and proofs were developed and
validated by the authors. They take full
responsibility for the content of the publication.

%=================================================================

\footnotesize
\bibliographystyle{plain}   % \citep
%\bibliography{references,hidyve,jp,relatedwork}
\bibliography{reference_sfsm}
\normalsize
% =======================================================================
\clearpage
\appendix
\section{Further algorithms for representative input sets}\label{app:algorithms}
\setcounter{algorithm}{0}
\renewcommand{\thealgorithm}{A.\arabic{algorithm}}
\providecommand{\theHalgorithm}{\arabic{algorithm}}
\renewcommand{\theHalgorithm}{appendix.\arabic{algorithm}}

This appendix gives the detailed constructions complementing Algorithm~\ref{alg:SigmaI-basic}. The general SMT variant also records requirement index sets for subsequent set-cover reduction. Further variants use MaxSMT or exploit input equivalence classes, a fixed reference model, a pair of models, or constant output assignments. These constructions are not implemented or evaluated experimentally in this paper.

Section~\ref{ssec:smt} specifies the solver interfaces, and Section~\ref{ssec:setcover} gives the general SMT and MaxSMT constructions and the set-cover procedure. Sections~\ref{ssec:eq}--\ref{ssec:mm} give the specialised constructions. Section~\ref{ssec:cont} treats constant output assignments.

\subsection{Interfaces to SMT solvers}\label{ssec:smt}

The algorithms below use the two solver interfaces defined in this section. We assume that every generated formula belongs to a theory for which the solver is guaranteed to return either $\textsf{sat}$ or $\textsf{unsat}$; an $\textsf{unknown}$ result is not considered.  For MaxSMT calls, we also assume that the returned solution is globally optimal. These interfaces can be implemented  with solvers such as Z3~\cite{MouraB08},
%{10.1007/978-3-540-78800-3_24}, 
OptiMathSAT~\cite{SebastianiT20}
(if linear arithmetic suffices), and Yices~\cite{Dutertre:cav2014}, chosen depending on the required theory and optimisation support.

%==========================================
 
%=====================================

% ............................................................
\subsubsection{SMT solver interface}

The basic interface receives a  quantifier-free formula $\expr$ over the input variables $x=(\chi_1,\dots,\chi_k)$. If $\expr$ is satisfiable, the solver returns $\textsf{sat}$ and a value tuple $a\in \DI$ satisfying $\expr(a)$. Otherwise, it returns $\textsf{unsat}$ and $a$ is undefined. Algorithm~\ref{alg:smt} specifies this interface.

\begin{algorithm} 
\caption{Function $\smta(\expr:\textit{FirstOrderFormula}) : (\mathbb B \times (\DI\cup \{\emptyset\}))$}
\label{alg:smt}
\begin{algorithmic}[1]
\REQUIRE $\expr$ : a quantifier-free formula over the input variables. 
\ENSURE $sat$ : a Boolean value that is true if and only if $\expr$ is satisfiable. 
\ENSURE $a$ : If $sat$ is $\ttt$, a value  tuple $a\in\DI$ satisfying $\expr(a)$; otherwise $a$ is undefined.
\end{algorithmic}
\end{algorithm}

% ............................................................

\subsubsection{MaxSMT solver interface}

Some algorithms use an \emph{unweighted partial MaxSMT solver}. The input consists of a hard constraint $\expr$ and a finite set $T$ of \emph{soft constraints}. If $\expr$ is satisfiable, the solver returns $\textsf{sat}$, together with a value tuple $a\in \DI$ satisfying a maximum number of constraints from $T$, and the set $T'=\{s\in T\mid s(a)\}\subseteq T$ of soft constraints satisfied by $a$. Formally, \[a\in \operatorname{arg\,max}_{b\in \DI: \expr(b)}\left|\{ s\in T\mid s(b)\}\right|.\]
If several arguments attain the same maximum, any one of them may be returned. If $\expr$ is unsatisfiable, the solver returns $\textsf{unsat}$; $a$ is undefined and $T'=\emptyset$.

\begin{algorithm}
\caption{Function\newline $\smsmt(\expr : \textit{FirstOrderFormula},T:\pwr(\textit{FirstOrderFormula})) :  (\mathbb B \times (\DI\cup \{\emptyset\}) \times \pwr(\textit{FirstOrderFormula}))$}
\label{alg:smsmt}
\begin{algorithmic}[1]
\REQUIRE $\expr$ : a hard constraint that every returned solution must satisfy. 
\REQUIRE $T$ : a set of soft constraints; the solver maximises the number satisfied by the returned solution.
\ENSURE $sat$ :  a Boolean value that is true if and only if $\expr$ is satisfiable.
\ENSURE $a$ : If $sat$ is $\ttt$, a model of $\expr$ that satisfies a maximum number of  constraints from $T$. Otherwise  undefined.
\ENSURE $T'\subseteq T$ : the soft constraints satisfied by $a$. If $sat = \fff$ then $T'=\emptyset$.  
\end{algorithmic}
\end{algorithm}
\subsection{Set-cover formulation for representative inputs}\label{ssec:setcover}
Let $G=\{g_1,\dots,g_r\}$ and $E=\{e_1,\dots,e_t\}$ be finite sets of guards and output assignments. To construct a finite input set $\Sigma_I$  representative for $(G,E)$, we introduce indices for the  guard-overlap  and output-separation requirements.    
\[
\Delta_G=\{(i,j)\mid 1\le i\le j\le r,\ g_i\cap g_j\neq\emptyset\}.
\]
Each pair $(i,j)\in\Delta_G$ denotes the requirement to select some input $a\in g_i\cap g_j$.
\[
\Delta_O=\{(i,j,p,q)\mid 1\le i\le j\le r,\ 1\le p<q\le t,
\exists a\in g_i\cap g_j:\ e_p(a)\neq e_q(a)\}.
\]
Each tuple $(i,j,p,q)\in\Delta_O$ denotes the requirement to select some input $a\in g_i\cap g_j$ such that $e_p(a)\neq e_q(a)$.

Since every output-separation requirement already implies a non-empty guard overlap, the indices of pure guard overlap requirements can be restricted to
\[
\Delta_G'=\Delta_G\setminus\{(i,j)\mid \exists p,q:\ (i,j,p,q)\in\Delta_O\}.
\]
Let $\Delta=\Delta_G'\cup\Delta_O$.

An SMT solver can construct the requirement sets together with a finite  \emph{candidate set} $S\subseteq\DI$.  For every non-empty guard overlap, the solver provides an input in that overlap. For every satisfiable output-separation constraint,
it provides an input that separates the corresponding output assignments. For $a\in S$, the \emph{cover function $C$} records all requirements satisfied by $a$: 
\begin{equation}\label{eq:coverC}
C(a)=\{(i,j)\in\Delta_G'\mid a\in g_i\cap g_j\}
\cup
\{(i,j,p,q)\in\Delta_O\mid a\in g_i\cap g_j\wedge e_p(a)\neq e_q(a)\} \subseteq \Delta.
\end{equation}
Any subset $\Sigma_I$ of $S$ is  representative for $(G,E)$ if it satisfies
\[
\Delta\subseteq\bigcup_{a\in\Sigma_I}C(a).
\]

Given the candidate set $S$, selecting a smallest subset $\Sigma_I\subseteq S$ that covers all requirements in $\Delta$ is an instance of the \emph{minimum set cover} problem. Its decision version is NP-complete, and the corresponding optimisation problem is NP-hard.  Algorithms~\ref{alg:SigmaI-build} and~\ref{alg:SigmaI-buildmax} generate a candidate set $S$.  Algorithm~\ref{alg:SigmaI-setcover} then selects a subset that still covers every requirement. The returned set need not have minimum cardinality, either within $S$ or over the full input domain $\DI$.

Algorithm~\ref{alg:SigmaI-build} uses only the SMT interface $\smta(\expr)$. For each guard pair $(g_i,g_j)$, it first checks whether the overlap is non-empty. If so, it checks every pair of output assignments separately and records a witness whenever the corresponding separation formula is satisfiable. The resulting candidate set covers all requirements.  

The SMT constraints are generated separately for each guard pair and each output-assignment pair. This algorithm does not try to maximise the number of requirements  covered by a single input. Algorithm~\ref{alg:SigmaI-setcover} takes this additional coverage into account. 

\begin{algorithm}[H]
\caption{Construct requirement index sets and a candidate set for $(G,E)$ with simple SMT interface, \newline  Function $\textsc{ConstructDeltaSMT}(G:\pwr(\textit{Guards}),E:\pwr(\textit{OutputAssignments})) : (\pwr(\DI)\times \pwr(\Nat^2)\times\pwr(\Nat^4)\times\pwr(\Nat^2\cup\Nat^4))$}
\label{alg:SigmaI-build}
\begin{algorithmic}[1]
\REQUIRE Guards $G=\{g_1,\dots,g_r\}$; output assignments $E=\{e_1,\dots,e_t\}$.
\ENSURE Candidate set $S\subseteq\DI$, and requirement index sets $\Delta_G',\Delta_O, \Delta=\Delta_G'\cup\Delta_O$.
\STATE $S\gets\emptyset$; $\Delta_G'\gets\emptyset$; $\Delta_O\gets\emptyset$
\FOR{$1\le i\le j \le r$}
    \STATE $(\mathsf{sat},a_{ij})\gets \smta(g_i\wedge g_j)$
    \IF{$\mathsf{sat}$}
      \STATE $\mathsf{sep}\gets\textbf{false}$
      \FOR{$1 \le p < q \le t$}
          \STATE $(\mathsf{sat},a)\gets \smta(g_i\wedge g_j\wedge(e_p(x)\neq e_q(x)))$
          \IF{$\mathsf{sat}$}
            \STATE $\Delta_O\gets\Delta_O\cup\{(i,j,p,q)\}$; $S\gets S\cup\{a\}$; $\mathsf{sep}\gets\textbf{true}$
          \ENDIF
      \ENDFOR
      \IF{$\neg\mathsf{sep}$}
        \STATE $\Delta_G'\gets\Delta_G'\cup\{(i,j)\}$; $S\gets S\cup\{a_{ij}\}$
      \ENDIF
    \ENDIF
\ENDFOR
\STATE $\Delta\gets\Delta_G'\cup\Delta_O$
\STATE \textbf{return} $(S,\Delta_G',\Delta_O,\Delta)$
\end{algorithmic}
\end{algorithm}

Algorithm~\ref{alg:SigmaI-setcover} starts with the set $U=\Delta$ of uncovered requirements. In each iteration, it selects a candidate that covers the maximum number of requirements still in $U$. The selected candidate is added to $\Sigma_I$ and the requirements it covers are removed from $U$.  The precondition
\[\Delta\subseteq \bigcup_{a\in S}C(a)\]
ensures that the loop terminates with $U=\emptyset$. All unselected candidates are discarded.

\begin{algorithm}[H]
\caption{Greedy set-cover selection of $\Sigma_I$, \newline
Function $\textsc{GreedySetCover}(S:\pwr(\DI),\Delta:\pwr(\Nat^2\cup\Nat^4), C: S \to 2^\Delta) : \pwr(\DI)$}
\label{alg:SigmaI-setcover}
\begin{algorithmic}[1]
\REQUIRE Candidate set $S\subseteq\DI$; requirement index set $\Delta \subseteq\bigcup_{a\in S}C(a)$  with a cover function $C(a)$ 
%defined in Equation~\eqref{eq:coverC}.
\ENSURE $\Sigma_I\subseteq S$ such that $\Delta\subseteq\bigcup_{a\in\Sigma_I}C(a)$.
\STATE $\Sigma_I\gets\emptyset$; $U\gets\Delta$
\STATE \textbf{for each} $a\in S$ \textbf{do} compute $C(a)$ \textbf{end for}
\WHILE{$U\neq\emptyset$}
  \STATE choose $a^\ast\in S$ maximising $|C(a^\ast)\cap U|$
  % \IF{$C(a^\ast)\cap U=\varnothing$}
  %   \STATE \textbf{break}   \BST{how can $C(a^\ast)\cap U=\varnothing$ become true in line~5?}
  % \ENDIF
  \STATE $\Sigma_I\gets\Sigma_I\cup\{a^\ast\}$; $U\gets U\setminus C(a^\ast)$; $S\gets S\setminus\{a^\ast\}$
\ENDWHILE
\STATE \textbf{return} $\Sigma_I$
\end{algorithmic}
\end{algorithm}

Algorithm~\ref{alg:SigmaI-buildmax} is an alternative to
 Algorithm~\ref{alg:SigmaI-build} that uses MaxSMT. For each non-empty guard overlap, it searches for an input satisfying as many unsolved separation constraints as possible.
 The requirements corresponding to $T'$ are marked as covered and removed from $T$. The loop ends when either $T=\emptyset$ or $T' = \emptyset$. In the first case, all separation constraints for the current overlap have been covered. In the second case, none of the remaining constraints is satisfiable on the overlap.

% \begin{algorithm}[H]
% \caption{Construct requirements index sets and candidate set for $(G,E)$ using MAX SMT techniques,\newline
% Function $\textsc{ConstructDeltaMaxSMT}(G,E) : (S,\Delta_G',\Delta_O,\Delta)$}
% \label{alg:SigmaI-buildmax}
% \begin{algorithmic}[1]
% \REQUIRE Guards $G=\{g_1,\dots,g_m\}$; output expressions $E=\{e_1,\dots,e_t\}$.
% \ENSURE Candidate set $S\subseteq\DI$ and requirements index set $\Delta=\Delta_G'\cup\Delta_O$.
% \STATE $S\gets\varnothing$; $\Delta_G'\gets\varnothing$; $\Delta_O\gets\varnothing$
% \STATE $U \gets \{ e_i \neq e_j~|~1\le i < j \le t\wedge e_i,e_j\in E  \}$ \COMMENT{Set of soft constraints}
% \STATE $T \gets U$
% \FOR{$1\le i\le j \le m$}
%     \STATE $(\mathsf{sat},a,T')\gets \smsmt(g_i\wedge g_j,T)$
%     \IF{$\mathsf{sat}$}
%        \STATE $S\gets S\cup \{ a \}$
%         \IF{$T'\neq\varnothing$}
%            \REPEAT
%            \STATE $\Delta_O \gets \Delta_O\cup \{ (i,j,p,q)~|~e_p\neq e_q \in T' \}$
%            \STATE $T \gets T\setminus T'$
%            \STATE \textbf{if} $T = \varnothing$ \textbf{then break}
%            \STATE $(\mathsf{sat},a,T')\gets \smsmt(g_i\wedge g_j,T)$
%            \STATE \textbf{if} $T' = \varnothing$ \textbf{then break} 
%            \STATE $S\gets S\cup \{ a \}$
%            \UNTIL{$\ttt$}
%            \STATE $T\gets U$
%         \ELSE
%            \STATE $\Delta_{G'} \gets \Delta_{G'} \cup \{ (i,j) \}$
%         \ENDIF
%     \ENDIF
% \ENDFOR
% \STATE $\Delta\gets\Delta_G'\cup\Delta_O$
% \STATE \textbf{return} $(S,\Delta_G',\Delta_O,\Delta)$
% \end{algorithmic}
% \end{algorithm}

\begin{algorithm}[H]
\caption{Construct requirement index sets and a candidate set for $(G,E)$ using MaxSMT techniques, \newline Function
$\textsc{ConstructDeltaMaxSMT}(G : \pwr(\textit{Guards}),E:\pwr(\textit{OutputAssignments})) : (\pwr(\DI)\times \pwr(\Nat^2)\times\pwr(\Nat^4)\times\pwr(\Nat^2\cup\Nat^4))$}
\label{alg:SigmaI-buildmax}
\begin{algorithmic}[1]
\REQUIRE Guards $G=\{g_1,\dots,g_r\}$; output assignments $E=\{e_1,\dots,e_t\}$.
\ENSURE Candidate set $S\subseteq\DI$ and requirement index sets $\Delta_G',\Delta_O, \Delta=\Delta_G'\cup\Delta_O$.
\STATE $S\gets\emptyset$; $\Delta_G'\gets\emptyset$; $\Delta_O\gets\emptyset$
\STATE $U \gets \{ e_p(x) \neq e_q(x)~|~1\le p < q\le t\wedge e_p,e_q\in E  \}$ \COMMENT{Set of soft constraints}

\FOR{$1\le i\le j \le r$}
    \STATE $T \gets U$
    \STATE $(\mathsf{sat},a,T')\gets \smsmt(g_i\wedge g_j,T)$
    \IF{$\mathsf{sat}$} 
        \IF{$T'=\emptyset$} 
           \STATE $S\gets S\cup \{ a \}$
           \STATE $\Delta_{G}' \gets \Delta_{G}' \cup \{ (i,j) \}$
        \ELSE
           \REPEAT
            \STATE $S\gets S\cup \{ a \}$
            \STATE $\Delta_O \gets \Delta_O\cup \{ (i,j,p,q)~|~e_p(x)\neq e_q(x) \in T' \}$
            \STATE $T \gets T\setminus T'$
            \STATE \textbf{if} $T = \emptyset$ \textbf{then break}
            \STATE $(\mathsf{sat},a,T')\gets \smsmt(g_i\wedge g_j,T)$
           \UNTIL{$T' = \emptyset$}        
        \ENDIF
    \ENDIF
\ENDFOR
\STATE $\Delta\gets\Delta_G'\cup\Delta_O$
\STATE \textbf{return} $(S,\Delta_G',\Delta_O,\Delta)$
\end{algorithmic}
\end{algorithm}

% -------------------------------------------------
\subsection{Construction via input equivalence classes}\label{ssec:eq}

For model learning, the guards are the input equivalence classes $\olg $ induced by the branching conditions $G'$. Each input equivalence class $\varphi_U$ is defined by Equation~\eqref{eq:classphi}; only non-empty $\varphi_U$ are added to $\ol G$. At most $2^{|G'|}$ such classes can occur. 
Algorithm~\ref{alg:enumclasses} avoids enumerating
all subsets of $G'$ explicitly by maintaining  only satisfiable prefixes.

The construction can be viewed as an ordered binary decision tree. The root, at level zero,  is labelled by $\ttt$.
A node at level $i-1$ with label $\varphi$ has two possible children, labelled $\varphi\wedge \neg g_i$ and $\varphi\wedge g_i$. Unsatisfiable children are omitted.

Each frontier element
$(\varphi,a,U)$ stores a satisfiable
prefix formula $\varphi$, a witness $a\in\varphi$, and the set $U$ of
branching conditions occurring positively in $\varphi$. When $g_i$ is processed, the witness $a$ already proves that one of the two extensions is satisfiable. Only the other extension requires an SMT call. After the last level, the frontier contains  exactly the triples $(\varphi_U,a_U,U)$ for the input equivalence classes. Projecting this set onto its first component yields $\olg$.

\begin{algorithm}[H]
\caption{Construction of $\ol G$ using an ordered binary decision tree over $G'$, \newline
Function $\textsc{EnumerateClasses}(G':\textit{BranchingConditions}^*) : \pwr(\textit{Guards}\times\DI\times\pwr(BranchingConditions))$}
\label{alg:enumclasses}
\begin{algorithmic}[1]
\REQUIRE Ordered branching conditions $G'=\langle g_1,\dots,g_r\rangle$ extracted from the SUT via static analysis.
\ENSURE A finite set $\Phi$ containing exactly one triple $(\varphi_U,a_U,U)$ for each $U\subseteq G'$ with $\varphi_U\neq \emptyset$, where $a_U$ is a selected witness of $\varphi_U$.
% $\Phi=\{(\varphi_U,a_U,U)\mid  U\subseteq \{g_1,\dots,g_r\}, \varphi_U\,\, \text{is satisfiable}, a_U\in\varphi_U\}$.
%\STATE $\Phi \gets \varnothing$
\STATE choose any $a_0\in \DI$
\STATE $\Phi \gets \{(\ttt,a_0,\emptyset)\}$
\FOR{$1 \le i \le r$}
  \STATE $\Phi_\textsf{new} \gets \emptyset$
   \FOR{\textbf{each} $(\varphi,a,U)\in \Phi$}
     \IF{$a\in g_i$}      
       \STATE $\Phi_\textsf{new} \gets \Phi_\textsf{new}\cup \{(\varphi\wedge g_i, a,   U\cup\{g_i\})\}$
       \STATE $(\textsf{sat},b) \gets \smta(\varphi \wedge \neg g_i)$
       \IF{\textsf{sat}}
        \STATE $\Phi_\textsf{new} \gets \Phi_\textsf{new}\cup \{(\varphi\wedge \neg g_i, b, U)\}$
       \ENDIF
     \ELSE
        \STATE \COMMENT{$a\in\neg g_i$}
        \STATE $\Phi_\textsf{new} \gets \Phi_\textsf{new}\cup \{(\varphi\wedge \neg g_i, a, U)\}$
       \STATE $(\textsf{sat},b) \gets \smta(\varphi \wedge g_i)$
       \IF{\textsf{sat}}
        \STATE $\Phi_\textsf{new} \gets \Phi_\textsf{new}\cup \{(\varphi\wedge g_i, b, U\cup \{g_i\})\}$
        \ENDIF
     \ENDIF
     
   \ENDFOR
   \STATE $\Phi\gets \Phi_\textsf{new}$   
\ENDFOR
%\STATE $\Phi\gets \textsf{Frontier}$
\STATE \textbf{return} $\Phi$
\end{algorithmic}
\end{algorithm}

An AllSMT enumeration procedure~\cite{SpallittaSB25}
can be used 
to enumerate the satisfiable truth assignments to the branching conditions. An experimental comparison with Algorithm~\ref{alg:enumclasses} is left for future work. 

% As an alternative to  Algorithm~\ref{alg:enumclasses}, an AllSMT-style loop~\cite{SPALLITTA2025104346} that enumerates all satisfiable truth-patterns of the guards and returns one witness per non-empty class can be implemented. Neither the frontier calculation  method of Algorithm~\ref{alg:enumclasses} nor the AllSMT technique are universally superior to the other: their performance depends on the concrete structure of the formulae in $G'$. %~\cite{DBLP:journals/dam/GrooteZ03}.

% Since $\ol G$ partitions the input space, the guard overlap condition is trivially fulfilled if one input $a\in \varphi$ exists for every $\varphi \in \ol G$. This insight leads to a simplified cover function $C_\Phi(a)$ that just contains the pairs of output assignments distinguished by input $a$ which resides in exactly one class $\varphi_U\in \ol G$:
% \[C_\Phi(a)=\{(p,q)\mid e_p(a)\neq e_q(a), 1\le p<q\le t\}.\]

Because $\olg$ is a partition of $\DI$, distinct classes do not overlap. The guard-overlap requirement therefore reduces to selecting at least one input from each class. Within a fixed class, the cover function needs to record only the output-assignment pairs separated by an input:
\[C_\olg(a)=\{(p,q)\mid e_p(a)\neq e_q(a), 1\le p<q\le t\}.\]

% An alternative construction first computes the finite family of input equivalence classes
% \[
% \Phi=\left\{\varphi_U\neq\varnothing\mid
% \varphi_U\equiv\bigwedge_{g_i\in U}g_i\wedge\bigwedge_{g_i\notin U}\neg g_i,
% \ U\subseteq G\right\}.
% \]
% These equivalence classes form a partition of $\DI$. Every non-empty guard overlap is a union of equivalence classes. For each class $\varphi_U$ and output pair $(e_p,e_q)$, satisfiability of
% \[
% \varphi_U\wedge(e_p\neq e_q)
% \]
% decides whether the two output assignments are distinguishable on that class. Witnesses returned by the solver are collected as candidates and then passed to the greedy set-cover step.

Algorithm~\ref{alg:SigmaI-atoms} processes each non-empty input
equivalence class separately. It starts with the witness $a_0$ returned by Algorithm~\ref{alg:enumclasses} and uses MaxSMT to add inputs that cover the unsolved output-separation requirements in the same class. If no output-assignment pair is 
distinguishable on the class, $a_0$
is retained as its guard-overlap witness. The union of the
sets selected for the individual classes is returned as $\Sigma_I$.

\begin{algorithm}[H]
\caption{Construction of $\Sigma_I$ via input equivalence classes,\newline
Function $\textsc{ConstructSigma}(G' : \textit{BranchingConditions}^*,E:\pwr(\textit{OutputAssignments})) : \pwr(\DI)$}
\label{alg:SigmaI-atoms}
\begin{algorithmic}[1]
\REQUIRE Ordered branching conditions $G'=\langle g_1,\dots,g_r\rangle$; output assignments $E=\{e_1,\dots,e_t\}$.
\ENSURE A finite   set $\Sigma_I\subseteq\DI$ that is representative for $(\olg,E)$, where $\olg$ is the set of input equivalence classes constructed from $G'$.
\STATE $\Sigma_I\gets\emptyset$
\STATE $\Phi\gets\textsc{EnumerateClasses}(G')$ \COMMENT{$\Phi=\{(\varphi_U,a_U,U)\mid \varphi_U\neq\emptyset\}$, each with witness $a_U\in\varphi_U$}
\FORALL{$(\varphi,a_0,U)\in\Phi$}
  \STATE $S\gets\{a_0\}$
  \STATE $R\gets\{(p,q)\mid 1\le p<q\le t\}$
  \STATE compute $C_\olg(a_0)$;    
  \STATE $R\gets R\setminus C_\olg(a_0)$
  \WHILE{$R\neq\emptyset$}
  \STATE $T \gets \{ e_p(x) \neq e_q(x)~|~(p,q)\in R \}$
  \STATE $(\textsf{sat},a,T')\gets \smsmt(\varphi,T)$
  \STATE $\textbf{if}\ T' = \emptyset \ \textbf{then break}$
  \STATE compute $C_\olg(a)$
  \STATE $R \gets R\setminus C_\olg(a)$
  \STATE $S \gets S\cup \{a\}$
  
  \ENDWHILE
  
\STATE $\Delta
       \gets \bigcup_{a\in S} C_\olg(a)$
 \IF{$\Delta=\emptyset$}
        \STATE $\Sigma_\varphi\gets\{a_0\}$
    \ELSE
        \STATE $\Sigma_\varphi
          \gets\textsc{GreedySetCover}{(S,\Delta, C_\olg)}$
          \COMMENT{Algorithm~\ref{alg:SigmaI-setcover}}
    \ENDIF      
\STATE $\Sigma_I\gets\Sigma_I\cup\Sigma_\varphi$ 
\ENDFOR
\STATE \textbf{return} $\Sigma_I$
\end{algorithmic}
\end{algorithm}

%--------------------------------------------
\subsection{Model-specific representatives}\label{ssec:mge}

For MBT, the reference model $M$ is known.  We are also given finite sets $G$
and $E$ such that 
the SUT has an SFSM representation whose transitions use
only guards from $G$ and output assignments from $E$.  The transition labels of 
$M$ need not belong to $G\times E$. Therefore, the
model-specific construction compares labels occurring
in $M$  with  possible SUT  labels from $G\times E$.

Let
\[
G_h=\{g^M_1,\ldots,g^M_u\}
\]
be the set of guards occurring in \(M\), and let
\[
E_h
=
\{e\mid \exists g\colon (g,e)\in A_M\}
=
\{e^M_1,\ldots,e^M_s\}
\]
be the set of output assignments occurring in \(M\). Furthermore, let
\[
G=\{g_1,\ldots,g_r\},
\qquad
E=\{e_1,\ldots,e_t\}.
\]
We  define
\[
O_M(i)
=
\{p\mid(g^M_i,e^M_p)\in A_M\}.
\]

Thus, \(O_M(i)\) contains the indices of the output assignments 
occurring with guard $g^M_i$ in $M$. Since $g^M_i\in G_h$, the set is non-empty.

According to Definition~\ref{def:repMGE}, the relevant guard-overlap
requirements are indexed by
\[
\Delta_G^M
=
\left\{
(i,j)
\;\middle|\;
1\leq i\leq u,\;
1\leq j\leq r,\;
g^M_i\cap g_j\neq\emptyset
\right\}.
\]
The output-separation requirements are indexed by
\[
\Delta_O
=
\left\{
(i,j,p,q)
\;\middle|\;
\begin{array}{l}
1\leq i\leq u,\;1\leq j\leq r,\;
p\in O_M(i),\;1\leq q\leq t,\\[1mm]
\exists a\in g^M_i\cap g_j\colon
e^M_p(a)\neq e_q(a)
\end{array}
\right\}.
\]
An output-separation witness is also a witness of the corresponding
guard overlap. Therefore, the guard-overlap requirements can be
restricted to
\[
\Delta_G'
=
\Delta_G^M
\setminus
\left\{
(i,j)
\;\middle|\;
\exists p,q\colon(i,j,p,q)\in\Delta_O
\right\}.
\]
Let
\[
\Delta=\Delta_G'\cup\Delta_O.
\]
For a concrete input \(a\in D_I\), define the model-specific cover
function
\[
\begin{split}
C_M(a)
={}&
\left\{
(i,j)\in\Delta_G'
\;\middle|\;
a\in g^M_i\cap g_j
\right\}\\
&{}\cup
\left\{
(i,j,p,q)\in\Delta_O
\;\middle|\;
a\in g^M_i\cap g_j
\land e^M_p(a)\neq e_q(a)
\right\}.
\end{split}
\]
A subset $\Sigma_I\subseteq S$ is representative for $(M,G,E)$
if
\[
\Delta
\subseteq
\bigcup_{a\in \Sigma_I}C_M(a).
\]

The general
construction can be applied to \(G_h\cup G\) and \(E_h\cup E\). It would,
however, introduce requirements that are not needed by Definition~\ref{def:repMGE}. 
Algorithm~\ref{alg:SigmaI-greedy-model} constructs only
the requirements between a label occurring in \(M\) and a possible label
from \(G\times E\).

\begin{algorithm}[H]
\caption{Greedy construction of $\Sigma_I$ for a fixed model $M$,\newline
Function $\textsc{ConstructSigma}(M:\textit{SFSM}, G:\pwr(\textit{Guards}), E:\pwr(\textit{OutputAssignments})) : \pwr(\DI)$}
\label{alg:SigmaI-greedy-model}
\begin{algorithmic}[1]
\REQUIRE Model $M$ with labels $(g^{M}_i,e^{M}_p)$, $G_h=\{g^M_1,\dots, g^M_u\}, E_h=\{e^M_1,\dots, e^M_s\}$; guards $G=\{g_1,\dots,g_r\}$; output assignments $E=\{e_1,\dots,e_t\}$.
\ENSURE A finite representative set $\Sigma_I\subseteq\DI$ for $(M,G,E)$.
\FOR{$i=1$ \textbf{to} $u$}
  \STATE $O_M(i)\gets\{p\mid(g^M_i,e^M_p)\in A_M\}$
\ENDFOR
\STATE $S\gets\emptyset$; $\Delta_G'\gets\emptyset$; $\Delta_O\gets\emptyset$
\FOR{$1 \le i\le u$}
\FOR{$1 \le j\le r$}
       \STATE $T \gets \{ e^M_p(x) \neq e_q(x)~|~p\in O_M(i), 1\le q\le t\}$
      \STATE $(\mathsf{sat},a,T')\gets\smsmt(g^M_i\wedge g_j,T)$
      \STATE \textbf{if} $ \neg \mathsf{sat}$ \textbf{then continue}
      \IF{$T' = \emptyset$}
        \STATE $S\gets S\cup\{a\}$
        \STATE $\Delta_G'\gets\Delta_G'\cup\{(i,j)\}$
      \ELSE
        \REPEAT
          \STATE $S\gets S\cup \{a\}$
          % \STATE $\Delta_O\gets\Delta_O\cup\{(i,j,p,q)\mid(p,q)\in T'\}$
          \STATE $\Delta_O\gets\Delta_O\cup\{(i,j,p,q)\mid (e^M_p(x)\neq e_q(x))\in T'\}$
          \STATE $T\gets T\setminus T'$
          \STATE \textbf{if} $T = \emptyset$ \textbf{then break} 
          \STATE $(\textsf{sat},a,T') \gets \smsmt(g^M_i\wedge g_j,T) $
        \UNTIL{$T' = \emptyset$}
      \ENDIF
\ENDFOR
\ENDFOR
\STATE $\Delta\gets\Delta_G'\cup\Delta_O$
\STATE $\Sigma_I\gets \textsc{GreedySetCover}(S,\Delta,C_M)$ \COMMENT{Algorithm~\ref{alg:SigmaI-setcover}}
\STATE \textbf{return} $\Sigma_I$
\end{algorithmic}
\end{algorithm}

% -----------------------------------------------------------------------------------------------------------------------------------------------------
\subsection{Pair-specific representatives}\label{ssec:mm}

 When both the reference model $M$ and an SFSM $M'$ representing the SUT are known, Definition~\ref{def:repMM} requires only guard overlaps and output-separation requirements between labels occurring in the two models. Algorithm~\ref{alg:SigmaI-greedy-modelMMprime} constructs a representative input set for $(M,M')$. For $a\in \DI$, define 
\[C_{M,M'}(a)=\{(i,j)\in \Delta_G'\mid a\in g_i\cap g_j'\}\cup \{(i,j,p,q)\in \Delta_O\mid a\in g_i\cap g_j'\wedge e_p(a)\neq e_q'(a)\}.\]
For each $g_i\in G_h$, define
\[
O_M(i)=\{\,p\mid (g_i,e_p)\in A_M\,\}.
\]
For each $g_j'\in G_{h'}$, define
\[
O_{M'}(j)=\{\,q\mid (g_j',e_q')\in A_{M'}\,\}.
\]
These sets contain the indices of all output assignments occurring with the corresponding guards. Since $G_h$ and $G_{h'}$ contain only guards occurring in the transition relations, all these sets are non-empty. 

\begin{algorithm}[H]
\caption{Greedy construction of $\Sigma_I$ that is representative for $(M,M')$,\newline
Function $\textsc{ConstructSigma}(M,M':\textit{SFSM}) : \pwr(\DI)$}
\label{alg:SigmaI-greedy-modelMMprime}
\begin{algorithmic}[1]
\REQUIRE Models $M$ with transition labels $(g_i,e_p)$, guards $G_h=\{g_1,\dots,g_u\}$, output assignments $E=\{e_1,\dots,e_t\}$, and $M'$ with labels $(g_i',e_p')$, guards $G_{h'}=\{g_1',\dots,g_{u'}'\}$, output assignments $E'=\{e_1',\dots,e_{t'}'\}$.
\ENSURE A finite representative set $\Sigma_I\subseteq\DI$ for $(M,M')$.
\FOR {$ 1\le i\le u$}
 \STATE $O_M(i)\gets\{p\mid(g_i,e_p)\in A_M\}$
 \ENDFOR
\FOR {$ 1\le j\le u'$}
 \STATE $ O_{M'}(j)\gets\{q\mid(g_j',e_q')\in A_{M'}\}$
 \ENDFOR 
\STATE $S\gets\emptyset$; $\Delta_G'\gets\emptyset$; $\Delta_O\gets\emptyset$
\FOR{$1 \le i \le u$}
 \FOR{$1\le j\le u'$}
       \STATE $T \gets \{ e_p(x) \neq e_q'(x)~|~(p,q)\in O_M(i)\times O_{M'}(j) \}$
      \STATE $(\mathsf{sat},a,T')\gets\smsmt(g_i\wedge g_j',T)$
      \STATE $\textbf{if}\ \neg \mathsf{sat}\ \textbf{then continue}$  \COMMENT{The guards do not overlap}
      \IF{$T' = \emptyset$} 
        \STATE $S\gets S\cup \{a\}$ 
        \STATE $\Delta_G'\gets\Delta_G'\cup\{(i,j)\}$
      \ELSE %\COMMENT{Can distinguish at least one pair of output expressions $(e_p,e_q')$ on $g_i\cap g_j'$}
        %\STATE $W\gets\emptyset$
        \REPEAT
          \STATE $S\gets S\cup \{a\}$
          % \STATE $\Delta_O\gets\Delta_O\cup\{(i,j,p,q)\mid(p,q)\in T'\}$
          \STATE $\Delta_O\gets\Delta_O\cup\{(i,j,p,q)\mid (p,q)\in O_M(i)\times O_{M'}(j)\wedge (e_p(x)\neq e_q'(x))\in T'\}$
          \STATE $T\gets T\setminus T'$
          \STATE \textbf{if} $T = \emptyset$ \textbf{then break} 
          \STATE $(\textsf{sat},a,T') \gets \smsmt(g_i\wedge g_j',T) $
        \UNTIL{$T' = \emptyset$}
      \ENDIF
   \ENDFOR   
\ENDFOR
\STATE $\Delta\gets\Delta_G'\cup\Delta_O$
\STATE $\Sigma_I\gets \textsc{GreedySetCover}(S,\Delta, C_{M,M'})$ \COMMENT{Algorithm~\ref{alg:SigmaI-setcover}}
\STATE \textbf{return} $\Sigma_I$
\end{algorithmic}
\end{algorithm}

 For a fixed guard pair $(g_i,g_j')$, the set $T$ contains the separation constraints between the output assignments occurring with these guards. The first MaxSMT call also checks whether $ g_i\land g_j'$ is satisfiable. If it is not, the guard pair is skipped. If it is satisfiable and $T'= \emptyset$, the returned input is kept as a guard-overlap witness. 

 Otherwise, the algorithm adds the returned input to $S$, marks the requirements in $T'$ as covered and removes these constraints from $T$. The loop ends when $T=\emptyset$, meaning that all relevant separation requirements have been covered, or when $T'=\emptyset$, meaning that none of the remaining constraints is satisfiable on the current guard overlap.

After all guard pairs have been processed, $\textsc{GreedySetCover}$ removes candidates that are not needed to cover 
$ \Delta=\Delta_G'\cup\Delta_O$.
The resulting set $\Sigma_I\subseteq S$ is representative for $(M,M')$.

% ------------------------------------------------------------------------------------------------

\subsection{Optimisation for constant output assignments}\label{ssec:cont}
If every output assignment is constant, then any two assignments are either equal on $\DI$ or distinguishable by every input in $\DI$. Hence, any witness of a non-empty guard overlap  satisfies
every applicable output-separation requirement. Algorithm~\ref{alg:constant-outputs} therefore constructs witnesses
 only for non-empty guard overlaps.

\begin{algorithm}[H]
\caption{Requirements and candidates for constant outputs, \newline
Assumption: all assignments in $E$ are constant, \newline
Function $\textsc{ConstructSigma}(G : \pwr(\textit{Guards})) : \pwr(\DI)$.}
\label{alg:constant-outputs}
\begin{algorithmic}[1]
\REQUIRE Guards $G=\{g_1,\dots,g_r\}$.
\ENSURE A finite representative set $\Sigma_I\subseteq\DI$ for $(G,E)$.
\STATE $S\gets\emptyset$; $\Delta\gets\emptyset$
\FOR{$i=1$ \textbf{to} $r$}
  \FOR{$j=i$ \textbf{to} $r$}
    \STATE $(\mathsf{sat},a_{ij})\gets\smta(g_i\wedge g_j)$
    \IF{$\mathsf{sat}$}
      \STATE $\Delta\gets\Delta\cup\{(i,j)\}$; $S\gets S\cup\{a_{ij}\}$
    \ENDIF
  \ENDFOR
\ENDFOR
%\STATE \textbf{return} $(S,\Delta)$
\STATE $\Sigma_I\gets \textsc{GreedySetCover}(S,\Delta,C)$ \COMMENT{Algorithm~\ref{alg:SigmaI-setcover}}
\STATE \textbf{return} $\Sigma_I$
\end{algorithmic}
\end{algorithm}

The set $E$ is not an input parameter of Algorithm~\ref{alg:constant-outputs}, because the cover function contains only the guard-overlap requirements:
\[
C(a)=\{(i,j)\in\Delta\mid a\in g_i\cap g_j\}.
\]
The greedy set cover step can still remove redundant candidates. For example, an input in $g_1\cap g_2\cap g_3$ covers the guard overlaps $g_1\cap g_2$, $ g_2\cap g_3$ and $ g_1\cap g_3$. Other witnesses of these overlaps can be removed if they cover no remaining requirements. 

Such constant-output cases occur in control systems that map ranges of analogue sensor values to fixed discrete control commands.

\clearpage

% =======================================================================

\end{document}